\documentclass[11pt,draftcls,onecolumn]{IEEEtran}
\usepackage{color}
\usepackage{graphicx}
\usepackage{bm}
\usepackage{url}
\usepackage{flushend,mystyle}
\usepackage{amsfonts,amssymb,amsmath}
\allowdisplaybreaks
\centerfigcaptionstrue
\def\defeq{\stackrel{\triangle}{=}}
\def\prob{\mathbb{P}}
\def\hii{h_{ii}}
\def\hij{h_{ij}}
\def\hri{h_{ri}}
\def\hrj{h_{rj}}
\def\hir{h_{ir}}

\def\nj{n_j}

\def\xonei{x_{1i}}
\def\yonei{y_{1i}}
\def\none1{n_{1i}}
\def\noner{n_{1r}}
\def\xtwoi{x_{2i}}
\def\ytwoi{y_{2i}}
\def\ntwoi{n_{2i}}
\def\yoner{y_{1r}}

\title{Opportunistic Wireless Relay Networks: Diversity-Multiplexing Tradeoff}

\author{Mohamed Abouelseoud, {\em Student Member, IEEE} and Aria
  Nosratinia, {\em Fellow, IEEE}\thanks{The authors are with the
    Department of Electrical Engineering, The University of Texas at
    Dallas, Richardson, TX 75080 USA, E-mail:
    m.abolsoud@student.utdallas.edu, aria@utdallas.edu. This work was
    supported in part by the National Science Foundation under grant
    CNS-0435429 and by the THECB under grant 009741-0084-2007.}}

\begin{document}

\maketitle

\vspace*{-0.2in}

\begin{abstract}

This work studies several relay networks whose {\em
opportunistic} diversity-multiplexing tradeoff (DMT) has been
unknown. 
Opportunistic analysis has traditionally relied on independence
assumptions that break down in many interesting and useful network
topologies
. This paper develops
techniques that expand opportunistic analysis to a broader class of
networks, proposes new opportunistic methods for several network
geometries, and analyzes them in the high-SNR regime.
For each of the geometries studied in the paper, we analyze the
opportunistic DMT of several relay protocols, including
amplify-and-forward, decode-and-forward, compress-and-forward,
non-orthogonal amplify-forward, and dynamic decode-forward.
Among the highlights of the results: in a variety of multi-user
single-relay networks, simple selection strategies are developed and
shown to be DMT-optimal. It is shown that compress-forward relaying
achieves the DMT upper bound in the opportunistic multiple-access
relay channel as well as in the opportunistic $n\times n$ user network
with relay.  Other protocols, e.g. dynamic decode-forward, are shown
to be near optimal in several cases. 
Finite-precision feedback is analyzed for the
opportunistic multiple-access relay channel, the opportunistic
broadcast relay channel, and the opportunistic gateway channel, and is
shown to be almost as good as full channel state information.
\end{abstract}



\section{introduction}
Opportunistic communication is a method that at each time chooses the
best among multiple communication alternatives in a network. Multiuser
diversity~\cite{Knopp1995} is a prominent example: in multiple-access
channels under quasi-static fading, it is throughput-optimal to allow
the user with the best channel to transmit at each time, while all
other users remain silent.  Relay selection is another example of
opportunistic communication. An early analysis of relay selection
without transmit-side channel state information appeared
in~\cite{Hunter2004}.  Bletsas et
al~\cite{Bletsas2006,Bletsas2007a,Bletsas2007b} investigated
amplify-and-forward (AF) relay selection, followed by several other
works
including~\cite{Jing2009,Krikidis2008,Krikidis2009,Vaze2009}. Decode-and-forward
(DF) relay selection has also received
attention~\cite{Nosratinia2007,Fareed2009,Beres2008b,Lee2009a,Oechtering2008,Yi2008,Hwang2008}.
The diversity multiplexing tradeoff (DMT) for relay selection has been
investigated in a few works including~\cite{Tannious2008} for
addressing the multiplexing loss of DF relaying, and~\cite{Vaze2009}
for a combination of antenna selection and AF relay selection.

The literature on opportunistic relays, despite its rapid
growth, has focused on a relatively restricted set of conditions.
Broadly speaking, the scope of previous work has been on geometries
and protocols where node selection can be reduced to scalar
comparisons of statistically independent link gains (or simple scalar
functions thereof). For example, Decode-Forward (DF) relay selection
compares the relay-destination links of relays that have decoded the
source message. In the case of amplify-forward (AF) relaying, the
end-to-end SNR (or a proxy, e.g. in~\cite{Bletsas2006}) is used to
select relays, which is again a scalar comparison among independent
random variables.

This leaves open a significant set of problems for whose analysis the
existing approaches are insufficient. Among them one may name even
seemingly simple problems, e.g. the DMT of the orthogonal relay on/off
problem in the single-relay channel, which has been unsolved until now
(see Section~\ref{sec:SingleRelay}).

To shed light on the key difficulties, consider the example of the
opportunistic multiple-access relay channel
(Figure~\ref{fig:MARC1}). Two users transmit messages to a common
receiver with the assistance of a relay. During each transmission
interval either User~1 transmits and User~2 is silent, or vice
versa. The goal is to opportunistically choose the user that can
access the channel at a higher rate.  The main challenge in the
analysis of this system is twofold:

\begin{figure}
\begin{center}
\includegraphics{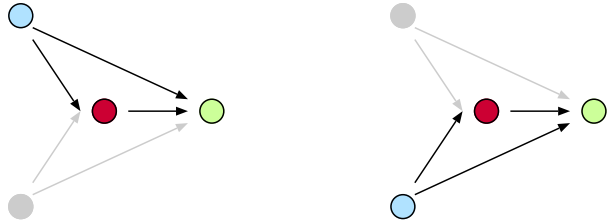}
\end{center}
\caption{The opportunistic modes in the multiple-access relay channel.}
\label{fig:MARC1}
\end{figure}

\begin{enumerate}
\item
The selection is a complex function of multiple link gains, i.e., it
is not immediately clear how to select the ``better'' node in an easy
and straight forward way. Not only do all the five link gains
participate in this decision, but also the capacity of the component
relay networks is generally unknown, and even the achievable rates are
only known as expressions that involve nontrivial
optimizations. Because the performance analysis must take into account
the selection function, the complexity of analysis can quickly get out
of hand with increasing number of nodes.
\item
The relay-destination link is shared among the two opportunistic
modes, therefore the decision variables for the two modes are not
statistically independent. The order statistics of dependent random
variables are complicated and often not computable in closed form.
\end{enumerate}

One of the contributions of this work is 
to address or circumvent the above mentioned difficulties.  This paper
analyzes the diversity and multiplexing gain of a variety of
opportunistic relay systems whose asymptotic high-SNR performance has
to date been unknown. All networks in this paper have one relay. Among
the network geometries that have been studied are the opportunistic
multiple-access and broadcast relay channels and several variations of
the opportunistic $n\times n$ user network with a relay. In the
$n\times n$ network with a relay, if nodes communicate pairwise while
crosslink gains cannot be ignored, the links and communication
structure resemble an interference channel with a relay, therefore we
call it an {\em opportunistic interference relay
  channel.}\footnote{The naming is for convenience purposes and only
  reflects the presence of links not the operation of the network. In
  opportunistic operation there is no interference among users.}  When
the crosslink gains can be ignored, we denote it the {\em
  opportunistic shared relay channel.}  Finally, if all transmitters
have data for all receivers, we denote the scenario as {\em
  opportunistic X-relay channel.} The {\em gateway channel} represents
a scenario where the only path between sources and destination is
through a relay.
%
%
To summarize, the main results of this paper are as follows:
\begin{itemize}
\item
To begin with, the DMT of the opportunistic single-relay on/off
problem is calculated under DF and AF. This simple result can be used
as a building block for the study of larger networks.

\item The diversity-multiplexing tradeoff of the opportunistic
  interference relay channel is calculated under orthogonal AF and DF,
  as well as non-orthogonal amplify and forward (NAF), dynamic decode
  and forward (DDF), and non-orthogonal compress and forward (CF). The
  nonorthogonal CF is shown to achieve the DMT upper bound.

\item For the shared relay channel, an upper bound for the DMT under
  opportunistic channel access is calculated. Furthermore, it is shown
  that for the shared relay channel at low multiplexing gain, the DDF
  outperforms the NAF and CF while at medium multiplexing gains, the
  CF is the best. At high multiplexing gain in the shared relay
  channel the relay should not be used.

\item For the multiple access relay channel, a simple selection scheme
  based on the source-destination link gains is shown to be optimal
  for several protocols.  Specifically, under this simple selection
  mechanism, the CF nonorthogonal relaying is shown to achieve the
  genie-aided DMT upper bound, and the NAF and the DDF also achieve
  their respective DMT upper bounds (i.e., more intricate selection
  schemes do not yield a better DMT).

\item For the X-relay channel, an opportunistic scheme is presented that
  meets the DMT upper bound under the CF protocol. For other
  relaying protocols, the DMT regions are calculated.

\item
  The results for the opportunistic broadcast relay channel follow
  from the opportunistic multiple-access relay channel.

\item For the gateway channel, the superposition as well as the
  orthogonal channel access is studied in the absence of transmit CSI,
  showing that the latter is almost as good as the former. Then, the
  opportunistic channel access is fully characterized.

\item
  Finite precision feedback is investigated for the multiple
  access relay channel (and by implication the broadcast relay channel),
  as well as the gateway channel. The DMT with finite-precision
  feedback for several other relay channels remains an open problem.
\end{itemize}

The organization of the paper is as follows: in
Section~\ref{sec:sys_mod}, we describe the system model. In
Section~\ref{sec:oprt}, the diversity multiplexing tradeoff for an
opportunistic system switches between different access modes is
analyzed. In Section~\ref{sec:SingleRelay}, the problem of a
single-relay opportunistic on/off problem is solved. Then, a
succession of DMT analyses is presented for a number of network
geometries and relaying protocols: in Section~\ref{sec:IRC} for the
interference relay channel, in Section~\ref{sec:SRC} for the shared
relay channel, in Section~\ref{sec:MARC} for the multiple access relay
channel, in Section~\ref{sec:XRC} for the X-relay channel, and in
Section~\ref{sec:GWC} for the gateway channel. We conclude our work in
Section~\ref{sec:conclusion}.


\section{System Model}
\label{sec:sys_mod}

All the nodes in the network are single-antenna and due to practical
limitations, nodes cannot transmit and receive at the same time (half
duplex). The channel between any two nodes experiences flat,
quasi-static block fading whose coefficients are known perfectly at
the receiver. The opportunistic selection mechanism also has access to
channel gains, either in full or quantized. The length of the fading
states (coherence length) is such that the source message is
transmitted and received within one coherence interval. Furthermore,
each transmission can accommodate a codeword of sufficient length so that
standard coding arguments apply.

The various networks considered in this paper may have either multiple
sources, multiple destinations, or both. In all scenarios in this
paper, there is one relay. The channel coefficients between
transmitter $i$ and receiver $j$ is denoted with $h_{ij}$. Channel
gains to or from a relay are shown with $\hir$ or $\hrj$. When the
network has only one source, a symbolic index $s$ is used for it;
similarly if a network has no more than one destination, the index $d$
will be used for it. For example, in a simple relay channel the links
are denoted $h_{sr},h_{rd},h_{sd}$. Channel gains are assumed
independent identically distributed circularly symmetric complex
Gaussian random variables.  The received signals are corrupted by
additive white Gaussian noise (AWGN) which is $n_r\sim{\cal{CN}}(0,N)$
at the relay and $\nj\sim{\cal{CN}}(0,N_j)$ at the
destinations. Without loss of generality, in the following we assume
all noises have unit variance, i.e., $N=N_j=1 \;\; \forall j$. The
transmitter nodes, the sources and the relay, have short-term
individual average power constraints for each transmitted
codeword. The transmit-equivalent signal-to-noise ratio (SNR) is
denoted by $\rho$. Due to the normalization of noise variance, the SNR
$\rho$ also serves as a proxy for transmit power.

In the original definitions of opportunistic communication,
e.g. multi-user diversity, only one transmitter is active during each
transmission interval. For the relay networks considered in this
paper, the definition is slightly generalized in the following manner:

\begin{definition}
Opportunistic communication is defined as a strategy where the
received signal at each node during each transmission interval is
independent of all but one of the transmitted messages. In other
words, during each transmission interval, each receiver in the network
hears only one message stream unencumbered by other message
streams. The target message stream may originate from a source, a
relay, or both.
\end{definition}

This definition maintains the spirit of opportunistic communication
while allowing various non-orthogonal relaying strategies. It is
noteworthy that with this generalized definition, in some networks
(e.g. shared relay channel) more than one message may be in transit at
the same time.

\begin{definition}
An {\em opportunistic communication mode} is the set of active
transmitters, receivers, and respective links in the network during a given
transmission interval.
\end{definition}

This work studies the high-SNR behavior of opportunistic relay
channels via the diversity-multiplexing tradeoff (DMT), in a manner
similar to~\cite{Zheng2003}. Each transmitter $i$ is allocated a
family of codes ${\cal{C}}_i(\rho)$ indexed by the SNR, $\rho$. The
rate $R_i(\rho)$ denotes the data rate in bits per second per hertz
and is a function of the SNR. The multiplexing gain per user $r_i$ is
defined as~\cite{Zheng2003}
\begin{equation}
r_i=\lim_{\rho\rightarrow \infty}\frac{R_i(\rho)}{\log\rho}.
\end{equation}
The selection strategy in the opportunistic relay network yields an
effective end-to-end channel. The attempted rate into this effective
channel is $R_i \approx r_i\log\rho$. The error probability subject to
this rate is denoted $P_e(\rho)$ and the diversity gain is defined as
follows.
\begin{equation}
d=-\lim_{\rho\rightarrow \infty}\frac{\log P_e(\rho)}{\log\rho},
\end{equation}
For the purposes of this study, since the transmission intervals are
sufficiently long, the diversity can be equivalently calculated using
the outage probability.

In principle, the high-SNR study of a network can generate a
multiplicity of diversities and multiplexing gains. In this paper we
pursue the symmetric case, i.e., all opportunistic modes the have the
same diversity gain $d$ (in a manner similar to~\cite{Tse2004}) and
also are required to support the same multiplexing gain $r_i$, where
$r_i=r/n$ and $r$ is the overall (sum) multiplexing gain.

Finally a few points regarding notation: The probability of an event is
denoted with $\prob(\cdot)$.  We say two functions $f(x)$ and $g(x)$
are exponentially equal if
\[\lim_{x\rightarrow\infty}\frac{\log
f(x)}{\log g(x)}=1 \; ,
\]
and denote it with $f(x)\doteq g(x)$. The exponential order of a
random variable $X$ with respect to SNR $\rho$ is defined as
\begin{equation}
  v=-\lim_{\rho\rightarrow\infty}\frac{\log X}{\log\rho},
\end{equation}
and denoted by $X\doteq\rho^{-v}$, $\dot\le$ and $\dot\ge$ follow the
same definition.


\section{Basic Results for DMT  Analysis}
\label{sec:oprt}

Consider an abstraction of a wireless network, shown in
Figure~\ref{fig:switch}, consisting of a set of sources, a set of
destinations, and a number of data-supporting paths between them. Each
of these paths may connect one or more source to one or more
destination, and may consist of active wireless links as well as
(possibly) relay nodes. Recall that the each collection of active paths
and nodes is called an {\em opportunistic mode.} A concrete example of
opportunistic modes was shown in Figure~\ref{fig:MARC1}, where Source~1,
Relay, Destination, and corresponding links make one mode, and Source~2,
Relay, Destination, and corresponding links form the second mode. For
the purposes of this section, the geometry of the links and relays that
compose each mode is abstracted away. However, the DMT supported by each
of the modes\footnote{The multiplexing gain of each mode can be defined
  as the prelog of the overall rate carried by that mode, and similarly
  the diversity defined as the slope of the corresponding aggregate
  error rate of the data.}  is assumed to be known. Furthermore, it is
assumed that only one mode can be active at any given time, i.e., we
select one mode during each transmission interval.
%
\begin{figure}[ht]
\centering
\includegraphics{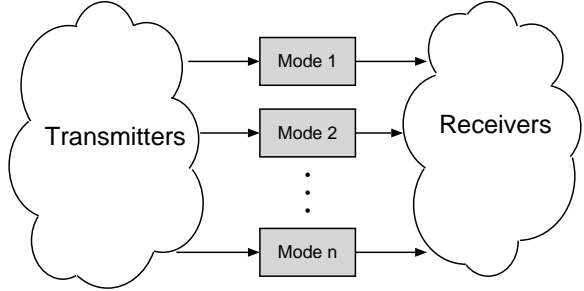}
\caption{General opportunistic wireless scenario model. Each mode
  consists of active links, potentially including a relay.}
\label{fig:switch}
\end{figure}

We now produce a simple but useful result.
\begin{lemma}
\label{lemma:1}
Consider a system that opportunistically switches between $n$ paths
(modes) whose conditional DMTs are given by $d'_i(r)$. The overall DMT is
bounded by:
\begin{equation}
d(r)\le d'_1(r)+d'_2(r)+\ldots+d'_n(r),
\end{equation}
where $d'_i(r)$ is defined as
\begin{equation}
\label{eq:dmt_dep}
d'_i(r)=-\lim_{\rho\rightarrow\infty}\frac{\log \prob(e_i
  |e_{i-1},\ldots,e_1 )}{\log\rho},
\end{equation}
\end{lemma}
and $\prob(e_i |e_{i-1},\ldots,e_1 )$ is the probability of error in
access mode $i$ given that all the previous access modes are in
error.

\begin{proof}
We demonstrate the result for a two-user network, generalization for $n$
users follows directly.

The total probability of error when switching between two subsystems is
\begin{equation}
P_e=\prob(e_1,e_2)+\prob(U_1,e_1,e_2^c)+\prob(U_2,e_1^c,e_2),
\end{equation}
where $e_1$ and $e_2$ are the events of error in decoding User~1 and User~2
data, respectively, the complements of error events are denoted with a
superscript $c$, and $U_1,U_2$ are the events of opportunistically
choosing User~1 and User~2, respectively. The event characterized by
the probabilities $\prob(U_1,e_1,e_2^c)$ and $\prob(U_2,e_1^c,e_2)$
represents the error due to wrong selection. 

We can upper bound $P_e$ as
\begin{align}
  P_e &\ge
  \prob(e_1,e_2)\nonumber\\ &=\prob(e_1)\prob(e_2|e_1)\nonumber\\&\doteq \rho^{-d_1'(r)}\rho^{-d_2'(r)},
\end{align}
which implies that
\begin{equation}
  d(r)\le d'_1(r)+d'_2(r),
\end{equation}

where $d'_i(r)$ is given by Equation~(\ref{eq:dmt_dep}). This
completes the proof.
\end{proof}

Specializing Lemma~\ref{lemma:1} to the case of independent error
probabilities directly yields the following.
\begin{lemma}
\label{lemma:2}
A DMT upper bound for opportunistically switching between $n$
\emph{independent} wireless subsystems is given by $d(r)$ where
\begin{equation}
d(r)\le d_1(r)+d_2(r)+\ldots+d_n(r),
\end{equation}
and $d_i(r)$ is the DMT of the subsystem $i$.
\end{lemma}

\begin{lemma}
\label{lemma:3}
The upper bounds of Lemma~\ref{lemma:1} and Lemma~\ref{lemma:2} are
tight if the following two conditions are asymptotically satisfied:

\begin{enumerate}
\itemsep 0pt
\item
Each selected subsystem uses codebooks that achieve its individual
DMT.
\item
The selection criterion is such that the system is in outage only when
all subsystems are in outage, i.e.,
$\prob(U_1,e_1,e_2^c)=\prob(U_2,e_1^c,e_2)=0$.

\end{enumerate}
\end{lemma}
Throughout the remainder of the paper, we assume that appropriate
codebooks are designed and used, therefore the first condition is
satisfied. The second condition would be satisfied by selecting access
modes according to their instantaneous end-to-end mutual
information. For practical reasons, we may consider simpler selection
criteria, in which case the tightness of the bounds above is not
automatically guaranteed. 


\section{Opportunistic On/Off Relay}
\label{sec:SingleRelay}

In this section we consider a simple orthogonal relaying scenario with
one source, one relay and one destination. During each transmission
interval, the source transmits during the first half-interval. In the
second half-interval, either the relay transmits, or the relay remains
silent and the source continues to transmit (see
Figure~\ref{fig:relayma}). The decision between these two options is
made opportunistically based on the channel gains.\footnote{Recall
  that both half-intervals are within the same coherence interval,
  i.e., the entire operation observes one set of channel
  realizations.}

The question is: how should the relay on/off decision be made, and
what is the resulting high-SNR performance (DMT). The apparent
simplicity of the problem can be deceiving, because the random
variables representing the performances of our two choices are not
independent.

\begin{theorem}
\label{theorem:ODMT}
The DMT of a three-node simple relay channel, under either AF or DF,
subject to opportunistic relay selection, is given by:
\begin{equation}
d(r)= (1-r)^+ + (1-2r)^+.
\end{equation}
\end{theorem}

\begin{proof}
The proof is  relegated to
Appendices~\ref{Appen:ODF} and~\ref{Appen:OAF}. An outline of the
proof is as follows. The DMT of a point-to-point non-relayed link is
$d(r)=(1-r)^+$. DF and AF orthogonal relaying~\cite{Laneman2004} have
the DMT $d(r)=(1-2r)^+$. Using the techniques described in the
previous section, these two DMTs are combined. The main part of the
proof is to establish that the {\em conditional} DMT of the relay channel
subject to the direct link being in outage is $d(r)=(1-2r)^+$, similar
to its unconditional DMT, therefore the overall result follows from Lemma~\ref{lemma:1}.
\end{proof}


\begin{figure}
\begin{center}
\includegraphics{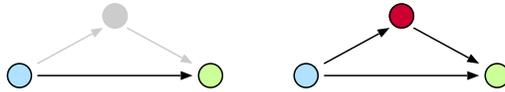}
\end{center}
\caption{The opportunistic modes in the simple orthogonal relay channel.}
\label{fig:relayma}
\end{figure}

\begin{remark}
For the simple relay channel shown above, there is no need to
investigate the opportunistic DDF and NAF, for the following
reason. In both NAF and DDF, it can be shown that the end-to-end
mutual information is never increased by removing the relay from the
network, because channel state information is already incorporated
into the operation of NAF and DDF in such a way that the usage of the
relay automatically adjusts to the quality of the links. 
\end{remark}

\begin{remark}
It has been known that the NAF protocol provides gains over orthogonal AF
but the NAF decoding can be complicated due to self-interference. The
results of this section show that the DMT gains of the NAF protocol
can be achieved with a much simpler decoding by using an
opportunistic relay on/off strategy. The cost is a small exchange of
channel state information for opportunistic relaying (1-bit feedback
from the destination node to the source and the relay).
\end{remark}


\section{Opportunistic Interference Relay Channel}
\label{sec:IRC}
This section is dedicated to the study of a $n\times n$ network with a
relay in the opportunistic mode. The topology of the links in this
network is identical to an interference relay channel, therefore this
structure is called an {\em opportunistic interference relay
  channel}. The naming is a device of convenience inspired by the
topology of the network.

For reference purposes, we briefly outline the background of {\em
  non-opportunistic} interference relay channel. The interference
channel~\cite{Carleial1978,Sato1981} together with a relay was
introduced by Sahin and Erkip~\cite{Sahin2007} (Figure~\ref{fig:IRC1})
who present achievable rates using full duplex relaying and rate
splitting. Sridharan et al.~\cite{Sridharan2008} present an achievable
rate region using a combination of the Han-Kobayashi coding scheme and
Costa's dirty paper coding, and calculate the degrees of
freedom. Maric et al.~\cite{Maric2008} study a special case where the
relay can observe the signal from only one source and forward the
interference to the other destination. Tannious and
Nosratinia~\cite{Tannious2008} show that the degrees of freedom for a
MIMO interference relay channel with number of antennas at the relay
matching or exceeding the number of users, is $k/2$ where $k$ is the
number of users.
\begin{figure}
\centering
\includegraphics{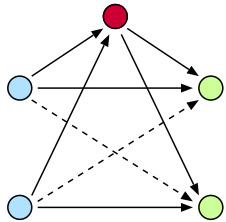}
\caption{Interference relay channel.}
\label{fig:IRC1}
\end{figure}

As mentioned earlier, opportunistic modes are defined such that the
data streams do not interfere, i.e., each receiving node is exposed to
one data stream at a time. Therefore, the two-user interference relay
channel has up to four access modes\footnote{In non-orthogonal CF,
  DDF, and NAF relaying protocols, the non-relayed modes never support
  higher rates than the relayed modes. Therefore in CF, DDF, NAF some
  of these modes are never selected and can be ignored.}  as shown in
Figure~\ref{fig:OIRC}.  The system selects one of the modes based on
the instantaneous link gains. In the following we analyze the network
under various relaying protocols and calculate the DMT in each case.

We start by developing a simple genie upper bound. Consider a genie
that provides the relay with perfect knowledge of the messages of the
transmitting sources. Thus access modes (c) and (d) are transformed
into a MISO channel with a DMT of $2(1-r)^+$. If the {\em genie-aided}
access mode (c) and (d) are in outage, then access modes (a) and (b)
will be in outage as well, therefore they need not be
considered. Applying Lemma~\ref{lemma:1}, the DMT of the $2\times2$
user opportunistic interference relay channel is upper bounded by
$4(1-r)^+$. This genie upper bound directly extends to $2n(1-r)^+$ for
the $n\times n$ user topology.

\subsection{Orthogonal Relaying}
Orthogonal relaying supports the full set of four access modes in
Figure~\ref{fig:OIRC}. Two of the modes do not involve the relay. In
the relay-assisted modes, a source transmits during the first half of
the transmission interval and the relay transmits in the second half
of the transmission interval.
\begin{figure*}
\centering
\includegraphics[width=6in]{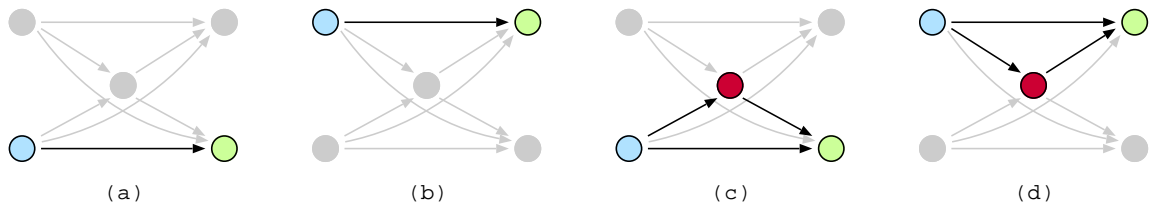}
\caption{The opportunistic access modes for the interference relay
  channel with orthogonal relaying.}
\label{fig:OIRC}
\end{figure*}
\subsubsection{Amplify and Forward Orthogonal Relaying}
In the relay-assisted modes, the relay amplifies the received signal
and forwards it to the destination. We select the mode that
minimizes the outage probability. The instantaneous mutual information
of the non-relay access modes is given by $I_i=\log(1+|h_{ii}|^2\rho)$
where $i=1,2$. The instantaneous mutual information for the
relay-assisted modes under orthogonal AF is given
by~\cite{Laneman2004,Hasna03}
\begin{equation}
\label{eq:mutaf1}
  I_{i+2}=\frac{1}{2} \log
  (1+|\hii |^2\rho+f(|\hir |^2\rho,|\hri |^2\rho)), \qquad i=1,2,
\end{equation}
where $f(x,y)=\frac{xy}{x+y+1}$. The selection criterion is as
follows. We first check the direct links. If none of the direct links
can support the rate $r\log\rho$, we check the access modes (c) and
(d). Using Lemma~\ref{lemma:1}, the total DMT is given by
\begin{align}
d(r)=d'_1(r)+d'_2(r)+d'_3(r)+d'_4(r),\label{eq:DMTOAF}
\end{align}
where
\begin{align}
  d'_1(r)=&\lim_{\rho\rightarrow\infty}\frac{\log\prob(e_1)}{\log\rho},&d'_2(r)=&\lim_{\rho\rightarrow\infty}\frac{\log\prob(e_2)}{\log\rho},\nonumber\\
  d'_3(r)=&\lim_{\rho\rightarrow\infty}\frac{\log\prob(e_3|e_1)}{\log\rho},
  &d'_4(r)=&\lim_{\rho\rightarrow\infty}\frac{\log\prob(e_4|e_2)}{\log\rho},\nonumber
\end{align}
It is easy to verify that $e_1$ and $e_3$ are independent from $e_2$
and $e_4$. Using techniques similar to the proof of 
Theorem~\ref{theorem:ODMT}, the outage
probability of the opportunistic orthogonal AF $2\times 2$
interference relay channel at high SNR is given by
\begin{align}
  \prob&(I<r\log
  \rho) \approx \bigg(\frac{e^{-2\rho^{2r-1}}-e^{-\rho^{r-1}}-e^{-2\rho^{2r-1}+\rho^{r-1}}+1}{1-e^{-\rho^{r-1}}}\bigg)^2(1-e^{-\rho^{r-1}})^2.\nonumber
\end{align}
The total DMT can be shown to be:
\begin{equation}
  d(r)=2(1-r)^++2(1-2r)^+.
\end{equation}
Generalization to $n$ source-destination pairs follows
easily; the corresponding DMT is $d(r)=n(1-r)^++n(1-2r)^+.$
\subsubsection{Decode and Forward Orthogonal Relaying}

We use the same selection technique used in the orthogonal AF
relaying. The instantaneous mutual information for the relay-assisted
modes is given by by~\cite{Laneman2004}
\begin{align}
  I_{i+2}&=\frac{1}{2}\log\big(1+\rho U_i\big), i=1,2
\end{align}
where 
\begin{equation}
  U_i=
  \begin{cases}
    2|h_{ii}|^2&|h_{ir}|^2<\frac{\rho^{2r}-1}{\rho}\\
    |h_{ii}|^2+|h_{ri}|^2&|h_{ir}|^2\ge\frac{\rho^{2r}-1}{\rho}
  \end{cases}
\end{equation}
With the same type of argument used to calculate the DMT for the
opportunistic orthogonal AF interference relay channel and
Appendix~\ref{Appen:ODF}, the outage probability of the opportunistic
orthogonal $2\times 2$ DF interference relay channel at high SNR is
given by
\begin{align}
\prob&(I<r\log \rho)\approx\bigg(1-e^{-\rho^{2r-1}}+\frac{ (1-
e^{-\rho^{r-1}}-\rho^{r-1}e^{-\rho^{2r-1}})e^{-\rho^{2r-1}}}{1- e^{-\rho^{r-1}}}\bigg)^2(1-e^{-\rho^{r-1}})^2.\nonumber
\end{align}
It can be shown that the DMT in case of orthogonal DF is
\begin{equation}
  d(r)=n(1-r)^++n(1-2r)^+.
\end{equation}

\subsection{Non orthogonal relaying}
In the non-orthogonal protocols considered in this section, the source
transmits throughout the transmission interval, while the relay
transmits during part of the transmission interval. The source and
relay signals are  superimposed at the destination. Note
that this superposition does not violate our working definition of
opportunistic communication, which states that the received signal at each
destination is independent of all but one of the transmitted
messages.


Under the non-orthogonal relaying protocols, the interference relay
channel has only two access modes, Figure~\ref{fig:OIRC} (c) and
(d). Access modes (a) and (b) are not considered, because it can be
shown that in non-orthogonal relaying, the end-to-end mutual
information of the relay-assisted modes is always greater than the
corresponding non-relayed modes.
\subsubsection{Non Orthogonal Amplify and Forward}
For half the transmission interval, the received signal at
the destination and at the relay are given
by~\cite{Nabar2004}
\begin{align}
\label{eq:AFout1}
\yonei=\sqrt{\rho}\, \hii \, \xonei+\none1 ,\;\;
\yoner=\sqrt{\rho}\; \hir \; \xonei+\noner ,\nonumber
\end{align}
The variables $x,y,n$ have two subscripts indicating the appropriate
half-interval and node identity, respectively. For example, $\yoner$
is the received signal during the first half-interval at the relay,
while $\xonei$ is the transmit signal at the first half-interval from
source $i$.  At the second half of the transmission interval the relay
normalizes the received signal (to satisfy the relay power constraint) and
retransmits it. The destination received signal in the second half is
given by
\begin{equation}
\label{eq:AFout2}
\ytwoi=\sqrt{\rho}\; \hii \; \xtwoi +\frac{\sqrt{\rho} \, \hri }{\sqrt{\rho|\hir |^2+1}}\; \yoner+\ntwoi ,\nonumber
\end{equation}
where a similar notation holds. 
The effective destination noise during this time is
$\frac{\sqrt{\rho}\hri }{\sqrt{\rho|\hir |^2+1}}\, \noner+\ntwoi $.

User $i^*$ is selected to maximize the mutual information, which at
high SNR can be shown to lead to the following selection rule:
\begin{align}
  i^*&=\arg\max_i I_i
   =\arg\max_i\bigg\{\frac{|\hii |^4|\hir |^2}{|\hri |^2+|\hir |^2}\bigg\},
\label{eq:NAFselection}
\end{align}

Using our knowledge of the DMT of non-opportunistic
NAF~\cite{Azarian2005} which is given by $d(r)=(1-r)^++(1-2r)^+$, and
applying Lemmas~\ref{lemma:2}, \ref{lemma:3} and using the selection
criterion $i^*$ from Equation~(\ref{eq:NAFselection}), the DMT of
opportunistic NAF interference relay channel with $n$
source-destination pairs is
\begin{equation}
d(r)=n(1-r)^++n(1-2r)^+.
\end{equation}
\subsubsection{ Dynamic Decode and Forward}
The relay listens to the source until it has enough information to
decode. The relay re-encodes the message using an independent Gaussian
codebook and transmits it during the remainder of the transmission
interval. The time needed for the relay to decode the message depends
on the quality of the source-relay channel.
Using~\cite{Azarian2005} and Lemma~\ref{lemma:1}, the DMT
of the optimal opportunistic DDF interference relay channel is as
follows:
\begin{equation}
  d(r)=
  \begin{cases}
    2n(1-r)  & 0\le r \le \frac{1}{2},\\
    n\frac{1-r}{r}&\frac{1}{2}< r \le 1.
  \end{cases}
\label{eq:DDF-upperbound}
\end{equation}

Compared to the other protocols considered for the interference relay
channel, the DDF mutual information for each node has a more complex
expression. This provides an impetus for the analysis of simpler
selection scenarios. It has been observed elsewhere in this paper that
selection based on source-destination link gains sometimes may perform
well, therefore we consider that choice function for the DDF
interference relay channel. Following the same technique
as~\cite{Abouelseoud2008B}, the resulting DMT can be shown to be
\begin{equation}
  d(r)=
  \begin{cases}
  (n+1)(1-r) &
   0\le r<\frac{n}{n+1}\\ n\frac{1-r}{r}&\frac{n}{n+1}\le r\le 1
  \end{cases}
\end{equation}
It is observed that for DDF, selection based on direct link gains is
clearly suboptimal, especially at low multiplexing gains.

\begin{figure*}[t]
\centering \includegraphics[width=4.5in]{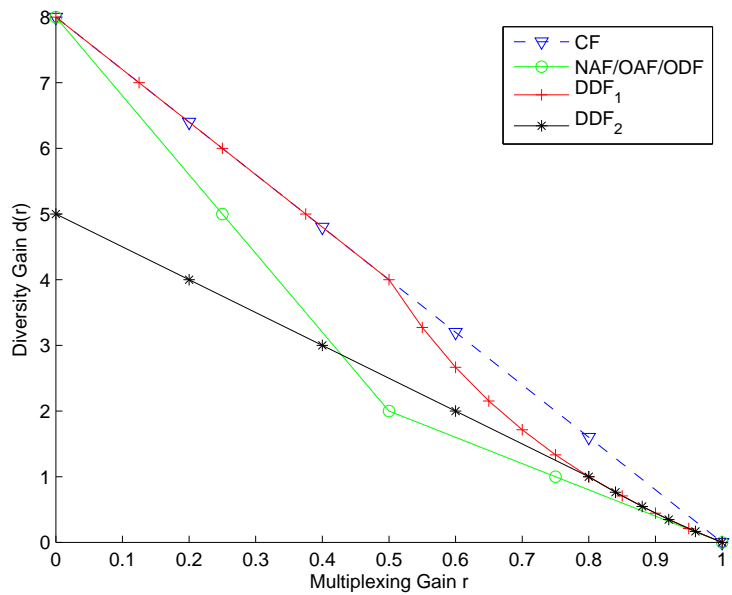}
\caption{Diversity multiplexing trade-off for a 4 source-destination
  pairs interference relay channel using different opportunistic
  relaying schemes.}
\label{fig:DMTIRC}
\end{figure*}
\subsubsection {Compress and Forward}
Following~\cite{Yuksel2007}, the relay listens to the selected source for a percentage $t$ of the
transmission interval. The source and the relay perform block Markov
superposition coding, and the destination employs backward
decoding~\cite{Kramer2005}. The relay performs Wyner-Ziv compression,
exploiting the destination's side information. This ensures that the
relay message can be received error free at the receiver. The relay
compression ratio must satisfy
\begin{equation}
I(\hat{y}_r;y_r|x_r, y_d)\le I(x_r;y_d).
\end{equation}
Yuksel and Erkip~\cite{Yuksel2007} show that the optimal DMT,
$d(r)=2(1-r)^+$, is achieved when the relay listens for half the
transmission interval and transmits during the remainder of time in
the interval\footnote{The work in~\cite{Yuksel2007} assumes transmit
  channel state information at the relay to insure that the relay's
  message reaches the destination error free. Recent
  work~\cite{Pawar2008} proves that the same DMT can be achieved using
  quantize-and-map relaying with only receiver-side channel state
  information. Another relaying protocol, dynamic
  compress-and-forward, is analyzed in~\cite{Gunduz2008} without a
  direct link and is shown to achieve the optimal DMT without channel
  state information at the relay.}.

For opportunistic compress and forward interference relay channel,
the user $i^*=\arg\max_i I_i$ is selected, where $I_i$ is the mutual
information for each access mode. At high-SNR, using results
from~\cite{Yuksel2007}, the selected user $i^*$ can be proved to be
\begin{align}
  i^*=&\arg\max_i \frac{(|h_{s_i,r}|^2+|h_{s_i,d_i}|^2)(|h_{r,d_i}|^2+
  |h_{s_i,d_i}|^2)|h_{s_i,d_i}|^2} {(|h_{s_i,r}|^2+|h_{s_i,d_i}|^2)
  +(|h_{r,d_i}|^2 +|h_{s_i,d_i}|^2)}.\nonumber
\end{align}
Each mode can achieve a DMT $d(r)=2(1-r)^+$, hence the opportunistic
system with $n$ source-destination pairs can achieve the DMT 
$d(r)=2n(1-r)^+$.

Figure~\ref{fig:DMTIRC} compares the DMT of various relaying schemes
for the interference relay channel with four source-destination
pairs. The optimal opportunistic DDF relaying is denoted by DDF1 and
DDF relaying with the simple selection criterion (based on
source-destination link gains) is denoted by DDF2. Compress and
forward achieves the optimal DMT but requires full CSI at the relay.

\section{Opportunistic Shared Relay Channel}
\label{sec:SRC}
The shared relay channel (SRC) (Figure~\ref{fig:SRC}) was introduced
in~\cite{Tajer2006a} with the sources using TDMA channel access and
orthogonal source and relay transmissions.
In~\cite{Khojastepour2008}, based on superposition and dirty paper
coding, lower and upper bounds on the capacity of additive white
Gaussian noise (AWGN) MIMO shared relay channel are presented.

\begin{figure}
\centering \includegraphics{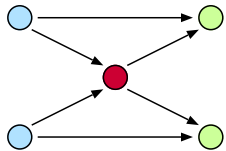}
\caption{Shared relay channel.}
\label{fig:SRC}
\end{figure}

In the shared relay channel, the direct link between each source
and its destination is free from interference from the other source,
however, the relay can cause indirect interference if it assists both
sources at the same time. Therefore, in the opportunistic mode the
relay should either assist one of the users or none of them
(Figure~\ref{fig:SRC3}). We assume the access mode that minimizes the
outage probability is chosen. In our analyses, access modes support
equal rate, thus in the first two access modes, one source transmits
at rate $R=r\log\rho$, while in the third access mode both sources
transmit, each with a rate $R_i=r/2\log\rho$.

\subsection{DMT Upper Bound}

An easy upper bound can be found by adding the DMT of the three access
modes without considering the dependencies among the throughputs of
the three access modes.
\begin{figure*}
\centering \includegraphics{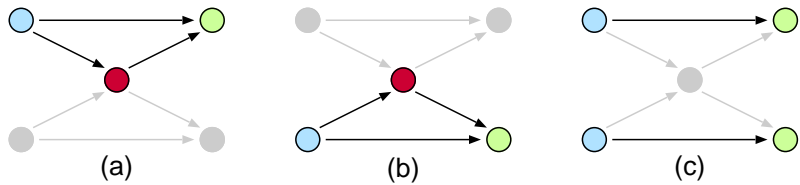}
\caption{Opportunistic access modes for the shared relay channel.}
\label{fig:SRC3}
\end{figure*}
A tighter upper bound can be found by assuming a genie that provides
the relay with the source information. In Figure~\ref{fig:SRC3}, we
call modes (a) and (b) {\em relay-assisted access modes} while
denoting mode (c) a {\em non-relayed access mode}. Thus, in the
presence of a genie, the relay-assisted access modes are essentially
equivalent to MISO links. The non-relay access mode is obviously not
affected by the genie. 
\begin{theorem}
  A DMT upper bound for the genie aided opportunistic shared relay channel
  is given by
\begin{align}
  d(r)& \le \big(1-\frac{r}{n}\big)^++(2n-1)(1-r)^+\\ &=
\begin{cases}
  2n-(2n-1+\frac{1}{n})r & 0\le r \le 1\\ (1-\frac{r}{n}) & 1 < r \le n
\end{cases}
\end{align}
\end{theorem}

\begin{proof}
The proof uses Lemma~\ref{lemma:1} taking into account the
dependency between the different access modes. Details of the proof are
given in Appendix~\ref{Appen:SRC2}.
\end{proof}
We notice that for high multiplexing gain, $r>1$, the first and second
access modes do not contribute to the diversity gain where the third
mode is always active. For low multiplexing gain, $r\le 1$, the three
access modes are contributing to the total diversity gain of the
system and switching between the three access modes should be
considered.

For clarity of exposition, we assume two source-destination pairs in
the remainder of the analysis. However, the analysis is extendable to
any number of node pairs in a manner that is straightforward.

\subsection{Achievable DMT}

If we allow ourselves to be guided by the upper bound
above, it is reasonable to use the non-relay access mode for high
multiplexing gains $(r>1)$. This makes intuitive sense, since relayed
access modes cannot support high multiplexing gains. For
multiplexing gains less than 1, switching between the three access mode
should be considered.

In the following we can consider a simplified selection by
partitioning the decision space: in one partition (at low multiplexing
gains) choosing only among relayed access modes (easier due to their
independence), in the other partition (at high multiplexing gains)
using only the non-relayed mode. This hybrid scheme sometimes is a
sufficient easy switching scheme and one can thus avoid the cost of
the comparison among all three modes especially in the cases of large
number of users. Using results from~\cite{Azarian2005}
and~\cite{Yuksel2007}, this
strategy leads to the following DMT for NAF
$d(r)=\max\Big\{2(1-r)^++2(1-2r)^+,\Big(1-\frac{r}{2}\Big)^+\Big\}$,
for DDF
\begin{align}
d(r)=&
\begin{cases}
    4(1-r) &0\le r\le 0.5\\
 2\frac{1-r}{r}& 0.5<r\le 3-\sqrt5\\
 1-\frac{r}{2} & 3-\sqrt5<r\le2.
  \end{cases}
\end{align} 
and for CF $d(r)=\max\Big\{4(1-r)^+,\Big(1-\frac{r}{2}\Big)^+\Big\}$.

Naturally, there is no guarantee that the above strategy is
optimal. For the best results, once must compare directly the three
opportunistic modes, but then the DMT requires nontrivial
calculations, as characterized by the following results.

The following DMT are subject to the two conditions mentioned in
Lemma~\ref{lemma:3}.

\subsubsection{Non-orthogonal Amplify and Forward}

\begin{theorem}
The overall DMT for the opportunistic shared relay channel under NAF
relaying protocol is given by
\begin{align}
  d(r)&=2(1-2r)^++(1-\frac{r}{2})^++(1-r)^+\nonumber\\
  &=\begin{cases}
    4-\frac{11}{2}r& 0\le r\le0.5\\
2-\frac{3}{2}r& 0.5<r\le 1\\
1-\frac{r}{2}&  1<r\le 2.
  \end{cases}
\end{align}
\end{theorem}
\begin{proof}
The proof uses Lemma~\ref{lemma:1} and results from MIMO
point to point communication~\cite{Zheng2003} and NAF
relaying~\cite{Azarian2005} taking into account the dependency between
the different access modes. Details are given in Appendix~\ref{Appen:SRC_NAF}.
\end{proof}

\subsubsection{Dynamic Decode and Forward}

\begin{theorem}
\label{theorem:SRC-DDF}
The overall DMT for the opportunistic shared relay channel under DDF relaying
protocol is given by
\begin{equation}
  d(r)=
  \begin{cases}
    \big(1-\frac{r}{1-r}\big(1-\frac{r}{2}\big)\big)+2(1-r)+
\big(1-\frac{r}{2}\big), & 0\le r \le 0.5\\

 2\frac{(1-r)}{r},& 0.5<r\le 2-\sqrt2\\

 \frac{(1-r)}{r}+\big(1-\frac{r}{2}\big),&2-\sqrt2 < r \le 1,\\
\big(1-\frac{r}{2}\big),&1 < r\le 2.
  \end{cases}
\end{equation}
\end{theorem}
\begin{proof}
The proof uses Lemma~\ref{lemma:1} and results for DDF
relaying~\cite{Azarian2005}, while taking into account the dependency
between the three access modes. Details are given in
Appendix~\ref{Appen:DDF}.
\end{proof}

\subsubsection{Compress and forward}
As mentioned earlier, the hybrid strategy yields the following DMT.
\begin{equation}
  d(r)=
  \begin{cases}
    4(1-r),&0\le r\le \frac{6}{7}\\ (1-\frac{r}{2}),& \frac{6}{7} <r \le 2
  \end{cases}
\end{equation}

One can show that optimization between all three access modes at each
$r$ cannot yield a better DMT under CF relaying, therefore the result
above cannot be improved upon. The proof is given in
Appendix~\ref{Appen:CF}.
\begin{figure*}
\centering \includegraphics{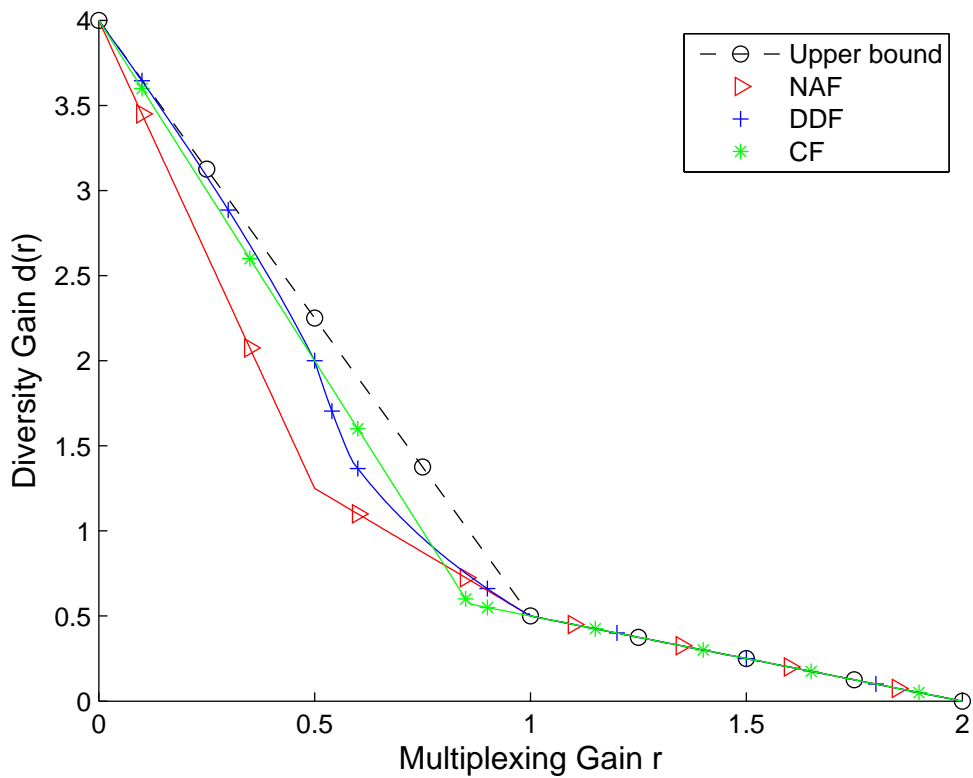}
\caption{Diversity multiplexing trade-off for a 2-pair shared relay
  channel, demonstrating the performance of various protocols.}
\label{fig:DMTSRC}
\end{figure*}
\begin{remark}
The trivial hybrid scheme of using the relay assisted modes at low
multiplexing gains and the direct links at high multiplexing links is
not always suboptimal. It is shown that for NAF and DDF, better
performance is achieved by considering the three access modes at low
multiplexing gains. For CF relaying, the non-relayed access mode is
not helping at low multiplexing gains, hence, the hybrid scheme
is optimal.
\end{remark}
\begin{remark}
Using the same technique used to prove the DMT of the orthogonal
opportunistic simple relay channel, Appendix~\ref{Appen:ODF}
and~\ref{Appen:OAF}, and Lemma~\ref{lemma:1}, one can show that the
DMT of the opportunistic shared relay channel under either orthogonal
AF or orthogonal DF is given by
\begin{equation}
d(r) = 2(1-2r)^+ + (1-r/2)^+,
\end{equation}
where the access modes are defined as before and the relay always
transmits orthogonal to the sources.
\end{remark}

To summarize the results for the opportunistic shared relay channel, a
brief comparison between three relaying protocols NAF, DDF, and CF is
as follows.  At low multiplexing gain the DDF outperforms NAF and CF.
At medium multiplexing gains, the relay does not have enough time to
fully forward the decoded message to the destination and the CF in
this case outperforms the DDF. At multiplexing gains above $1$, it
does not matter which relaying protocol is used since the
DMT-optimal strategy uses direct (non-relayed) mode.


\section{Opportunistic Multiple Access and Broadcast Relay Channels}
\label{sec:MARC}

\begin{figure}
\centering
\includegraphics{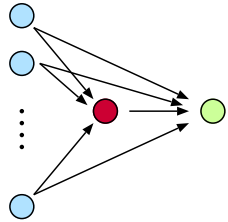}
\caption{The multiple access relay channel.}
\label{fig:marc}
\end{figure}

The multiple access relay channel (MARC)~\cite{Kramer2000} consists of
the standard multiple access channel together with one relay (see
Figure~\ref{fig:marc}).  No results for the DMT of the opportunistic
MARC have been available until now, but its {\em non-opportunistic}
DMT under superposition coding with single-antenna nodes is analyzed
in~\cite{Azarian2008,chen2006,chen2007,Yuksel2007}. The following
results are known for the non-opportunistic MARC: It is known that the
dynamic decode and forward is DMT optimal for low multiplexing
gain~\cite{Azarian2008}. The compress and forward protocol achieves a
significant portion of the half duplex DMT upper bound for high
multiplexing gain~\cite{Yuksel2007} but suffers from diversity loss in
the low multiplexing regime. The multiple-access relay amplify and
forward (MAF) is proposed in~\cite{chen2007}, it dominates the CF and
outperform the DDF protocol in high multiplexing regime.

The broadcast relay channel (BRC) was introduced independently
in~\cite{Kramer2004} and~\cite{Liang2004}.  Assuming single-antenna
nodes, the {\em opportunistic} BRC is identical to the {\em
  opportunistic} MARC save for certain practicalities in the exchange
of channel state information, which does not make a difference at the
abstraction level of the models used in this paper. Therefore for the
demonstration purposes we focus on MARC; the results carry over to the
BRC directly.

\subsection{DMT Upper Bound}

In order to calculate a DMT upper bound for the opportunistic MARC, we
assume a genie gives the relay an error-free version of the messages
originating from all the sources. We also assume full cooperation on
the transmit side. Under these conditions, the source that maximizes
the instantaneous end-to-end mutual information is selected. Each of
the $n$ sources has an independent link to the destination and they
all share the same relay-destination link. The opportunistic modes are
demonstrated in Figure~\ref{fig:MARC_AM}. The genie-aided MARC is
equivalent to a MISO system with $n+1$ transmit antennas and one
receive antenna.

The performance of the {\em opportunistic} genie-aided MARC is
therefore upper bounded by a $(n+1)\times 1$ MISO system with antenna
selection that chooses for each codeword transmission two transmit
antennas. The $(n+1)\times 1$ antenna selection allows configurations
that do not have a counterpart in the opportunistic modes in the MARC
channel, therefore due to the extra flexibility the MISO system with
antenna selection upper bounds the performance of the genie-aided
opportunistic MARC channel.

The DMT of a $M\times N$ MIMO link with $L_t<M$ selected transmit
antennas and $L_r<N$ selected receive antennas is upper bounded by a
piecewise linear function obtained by connecting the following $K+2$
points~\cite{Jiang2007}
\begin{equation}
\label{eq:Ant_sel}
  \Big\{\Big(n,(M_r-n)(M_t-n)\Big)\Big\}_{n=0}^{K},\big(\min(L_r,L_t),0\big),
\end{equation}
where
\begin{align*}
  K=&\arg\min_{k\in \mathbb
    Z}\frac{(M_r-k)(M_t-k)}{\min(L_r,L_t)-k},\\
  &\text{subject to } \;0\le k\le \min(L_r,L_t)-1
\end{align*}
Using this result, a $(n+1)\times 1$ MISO system with two selected
transmit antennas has a DMT that is upper bounded by
\begin{equation}
d(r)=(n+1)(1-r)^+.
\end{equation}
This represents our genie-aided upper bound for opportunistic MARC.
\begin{figure*}
\centering \includegraphics[width=5in]{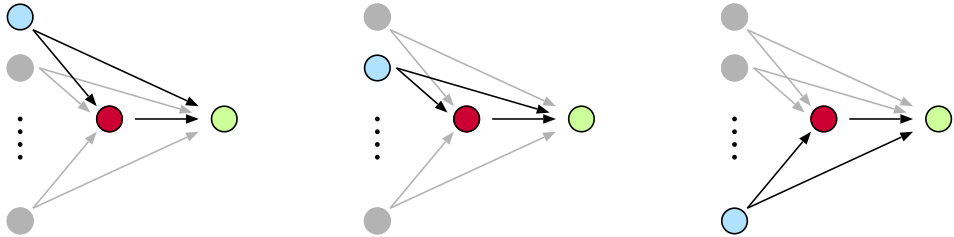}
\caption{Opportunistic access modes for the genie-aided multiple
  access relay channel.}
\label{fig:MARC_AM}
\end{figure*}

\subsection{Achievable DMT}

In this section, we propose a node selection rule and calculate the
corresponding achievability results for a number of relaying protocols
in opportunistic MARC and BRC. As mentioned earlier, one of the
difficulties in the computation of DMT in opportunistic scenarios is
the dependencies among the statistics of the node selections, which
itself is a result of selection rules. To circumvent these
difficulties, we propose a selection rule that relies only on the
source-destination links in the MARC. Because this method does not
observe the shared link in the system, the resulting node statistics
are independent and many of the computational difficulties
disappear. 

We shall see that this simplified selection works surprisingly well in
the high-SNR regime.
It will be shown that for some relaying protocols this selection
algorithm yields achievable DMT that is tight against the upper
bound. 


The proposed schemes for the MARC can be also be used for the BRC,
therefore for demonstration purposes we limit ourselves to MARC. The
only difference is that for the BRC the CSI must be fed back to the
source to make the scheduling decision.

\subsubsection{Orthogonal Amplify and Forward}

The maximum instantaneous mutual information between the inputs and
the output is
\begin{equation}
I_{AF}=\frac{1}{2} \log\big
(1+\rho|h_{{i^*}d}|^2+f(\rho|h_{{i^*}r}|^2,\rho|h_{rd}|^2)\big),
\end{equation}
where ${i^*}=\arg\max_i|h_{id}|$. The outage probability is given by
\begin{align}
P_{AF}&=\prob \big(I_{AF}<r\log\rho\big)\nonumber\\
        &=\prob \bigg(|h_{{i^*}d}|^2+\frac{1}{\rho}f(\rho|h_{{i^*}r}|^2,\rho|h_{rd}|^2)<\frac{\rho^{2r}-1}{\rho}\bigg)\label{eq:out_oaf}.
\end{align}
Since channel coefficients $\hij$ are complex Gaussian,
$|\hij|^2$ obey exponential distributions. We therefore use the
following result to characterize~(\ref{eq:out_oaf}) in the high-SNR
regime.

\begin{lemma}
\label{lemma:4}
Assume random variables $u_i$, $v$ and $w$ follow exponential distributions with
parameters $\lambda_u,\lambda_v$ and $\lambda_w$, respectively, and
$\epsilon$ is a constant and $f(x,y)=\frac{xy}{x+y+1}$.
\begin{align}
\label{eq:res1}
\lim_{\rho\rightarrow\infty}\frac{1}{\Big(\frac{\rho^{2r}-1}{\rho}\Big)}\,\prob \bigg(u_i<\frac{\rho^{2r}-1}{\rho}\bigg)&=\lambda_u\\
\label{eq:res2}
\lim_{\rho\rightarrow \infty}\frac{1}{\Big(\frac{\rho^{2r}-1}{\rho}\Big)^n}\,\prob \bigg(\max_i
u_i<\frac{\rho^{2r}-1}{\rho}\bigg)&=\lambda_u^n\\
 \label{eq:res3}
\lim_{\rho\rightarrow\infty}\frac{1}{\Big(\frac{\rho^{2r}-1}{\rho}\Big)^{n+1}}\,\prob \bigg(\max_i u_i+v<\frac{\rho^{2r}-1}{\rho}\bigg)&=\frac
{\lambda_v\lambda_u^n}{n+1},\\
\lim_{\rho\rightarrow \infty} \frac{1}{\Big(\frac{\rho^{2r}-1}{\rho}\Big)^{n+1}}\,\prob \bigg(\max_i
u_i+ f(\frac{v}{\epsilon},\frac{w}{\epsilon}) <\frac{\rho^{2r}-1}{\rho}\bigg)
&=\frac{\lambda_u^n(\lambda_v+\lambda_w)}{2},\label{eq:result1}
\end{align} 
\end{lemma}

\begin{proof}
Expression~(\ref{eq:res1}) is proved in~\cite{Laneman2004}. The proof of the
other expressions is similar (with slight modifications)
and is omitted for brevity.
\end{proof} 

From~(\ref{eq:out_oaf}) and~(\ref{eq:result1}),
the probability of outage at high SNR is
\begin{align}
P_{AF}&\doteq \frac{1}{2}\lambda_{i^*d}^n\big(\lambda_{i^*r}+\lambda_{rd}\big)\bigg(\frac{\rho^{2r}-1}{\rho}\bigg)^{n+1}
\end{align}
where $\lambda_{i^*r}, \lambda_{rd}, \lambda_{i^*d}$ are the
exponential parameters of the channel gains for the links
corresponding to the selected opportunistic mode. It follows that the
DMT of the opportunistic $n$-user MARC with orthogonal
amplify-and-forward, under a selection rule based on the
source-destination channel gain, is given by
\begin{equation}
d(r)=(n+1)(1-2r)^+.
\end{equation}

\subsubsection{Orthogonal Decode and Forward}

With the orthogonal DF protocol, outage happens if either of the
following two scenarios happen: (1) the relay cannot decode and the
direct source-destination channel is in outage, or (2) the relay can
decode but the source-destination and relay-destination links together
are not strong enough to support the required rate. In other words:
\begin{align}
P_{DF}&=\prob \bigg(|h_{{i^*}r}|^2\ge\frac{\rho^{2r}-1}{\rho}\bigg)\prob \bigg(
|h_{{i^*}d}|^2+|h_{rd}|^2<\frac{\rho^{2r}-1}{\rho}\bigg)\nonumber\\ 
&\qquad +\prob \bigg(|h_{{i^*}r}|^2<\frac{\rho^{2r}-1}{\rho}\bigg)\prob \bigg(
|h_{{i^*}d}|^2<\frac{\rho^{2r}-1}{\rho}\bigg)\nonumber\\&\doteq \prob \bigg(|h_{{i^*}r}|^2\ge\rho^{2r-1}\bigg)\prob \bigg(
|h_{{i^*}d}|^2+|h_{rd}|^2<\rho^{2r-1}\bigg)\nonumber\\ 
&\qquad+\prob \bigg(|h_{{i^*}r}|^2<\rho^{2r-1}\bigg)\prob \bigg(
|h_{{i^*}d}|^2<\rho^{2r-1}\bigg).
\end{align}

Using Lemma~\ref{lemma:4} (specifically
equations~(\ref{eq:res1}),~(\ref{eq:res2}),~(\ref{eq:res3})) the
outage probability can be approximated thus:
\begin{align}
P&_{DF}\doteq \bigg(\frac{\lambda_{i^*d}^n
\lambda_{rd}}{n+1}+\lambda_{i^*d}^n \lambda_{rd}\bigg)
\rho^{(n+1)(2r-1)}
\end{align}
It follows directly that the $n$-user opportunistic MARC,
subject to selection based on source-destination channel gains and
operating with orthogonal DF, has the following DMT
\begin{equation}
d(r)=(n+1)(1-2r)^+.
\end{equation}

\begin{remark}
We know that an orthogonal relay may not be helpful in high
multiplexing gains, but the above orthogonal MARC dedicates time to
the relay, therefore it may be improved. To do that, we add to the
system $n$ {\em unassisted} modes, where the relay does not play a
role. For an opportunistic MARC that can choose between $2n$
opportunistic modes, one can show that the maximum achieved DMT is
$d(r)=n(1-r)^++(1-r/2)^+$. A simple selection rule achieves this DMT:
take the best source-destination link. If it is viable without the
relay, use it without relay, otherwise use it with the relay.
\end{remark}

\subsubsection{Non-Orthogonal Amplify and Forward}
In this protocol, the source with the maximum source-destination
channel coefficient is selected. Recall that the index of this source
is denoted $i^*$. This source continues transmitting throughout the
transmission interval.


The DMT of the MARC with n sources and opportunistic channel access
{\em based on the source-destination channel gain} using NAF relaying
is
\begin{equation}
\label{eq:MDT_OAF}
d(r)=n(1-r)+(1-2r)^+.
\end{equation}
This result indicates that at multiplexing gains $r> 0.5$ the relay
does not play any role; the only available diversity at $r>0.5$ is
that of multiuser diversity generated by selection among $n$ sources.

To prove the result, we make use of the calculation method
in~\cite{Zheng2003,Azarian2005}.  An outline of the proof is as
follows. We assume that $v_1$ is the exponential order of the random
variable $\frac{1}{|h_{{i^*}d}|^2}$, i.e.
\begin{equation}
v_1=-\frac{\log(|h_{{i^*}d}|^2)}{\log{\rho}}.
\end{equation}
The probability density function of the exponential order is
\begin{equation}
p_v=n\ln(\rho) \rho^{-v}e^{-\rho^{-v}}(1-e^{-\rho^{-v}})^{n-1},
\end{equation}
which, asymptotically,
\begin{equation}
p_v\doteq \begin{cases} 0, & v<0 \\\rho^{-nv}, & v\ge0. \end{cases}
\end{equation}
The probability of outage can be characterized by $P_O\doteq
\rho^{-d_o}$ where
\begin{equation}
d_o=\inf_{(v_1,v_2,u) \in O^+}nv_1+v_2+u,
\end{equation}
where $v_2$ and $u$ are the exponential order of $1/|h_{{i^*}r}|^2$ and
$1/|h_{rd}|^2$, respectively. The set $O$ characterizes the outage event
and $O^+$ is $O\bigcap R^{3+}$. Optimization problems of this form
have been solved in~\cite{Azarian2005} and also in the context of
opportunistic relay networks we have demonstrated a solution in
Appendix~\ref{Appen:SRC_NAF} for the shared relay channel, therefore
we omit a similar solution here in the interest of brevity. 

\begin{figure*}
\centering
\includegraphics[width=5.5in]{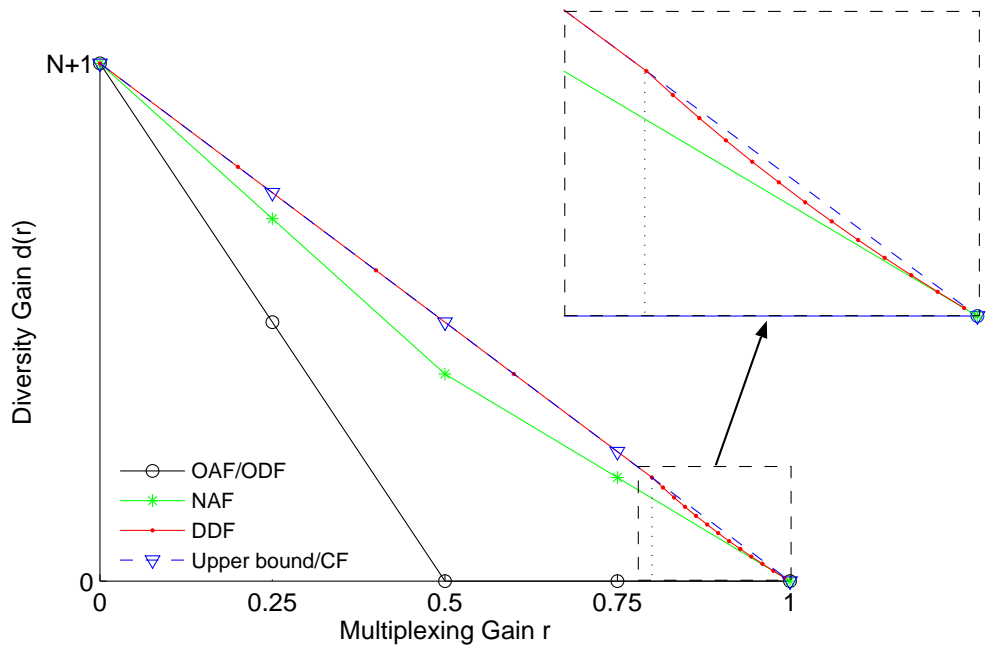}
\caption{DMT for a N-user opportunistic multiple-access relay
  channel. The insert shows the high-multiplexing gain region.}
\label{fig:DMT}
\end{figure*}

\subsubsection{Dynamic Decode and Forward}

The DMT of the opportunistic DDF MARC, where the selection is based on
the source-destination channel gain, is given by
\begin{equation}
d(r)=\begin{cases} (n+1)(1-r),&\frac{n}{n+1}\ge r\ge 0\\
n\frac{(1-r)}{r},&1\ge r \ge \frac{n}{n+1}.\end{cases}
\end{equation}
The proof, which is omitted for brevity,
follows~\cite{Zheng2003,Azarian2005} together with the basic Lemmas of
this paper and the NAF MARC proof. The DDF achieves the optimal trade-off (the genie-aided
DMT) for $\frac{n}{n+1}\ge r\ge 0$. For multiplexing gains
$r>\frac{n}{n+1}$ the relay does not have enough time to perfectly
help the selected source. However, as $n$ grows, the DMT approaches
the upper bound (genie-aided).

\subsubsection{Compress and Forward}
The node selected by the opportunistic
algorithm has index $i^*$. The system will be in outage if the
transmission rate $r\log\rho$ is less than the instantaneous mutual
information $I(x_{i^*};\hat{y_r},y_d|x_r)$, where $\hat{y_r}$
represents the compressed signal at the relay, $y_r$ and $y_d$ are the
received signals at the relay and the destination, respectively, and
$x_{i^*}$ and $x_r$ are the source and relay transmitted signals,
respectively.  Using selection scheme based on the direct link only
and applying the same techniques as in~\cite{Yuksel2007}, it
follows that the CF protocol achieves the following DMT
\begin{equation}
d(r)=\min\big(d_{BC}(r),d_{MAC}(r)\big),
\label{eq:DMTCF}
\end{equation}
where $d_{BC},d_{MAC}$ correspond to the outage of broadcast and MAC
cutsets, as follows:
\begin{align}
d_{BC}(r)& \defeq -\lim_{\rho\rightarrow\infty}\frac{\underset{p(x_{i^*},x_r)} 
{\min}\prob(I(x_{i^*};y_r y_d|x_r)<
r\log\rho)}{\log\rho}\nonumber\\
&=-\lim_{\rho\rightarrow\infty}\frac{\prob\big(\log\big|I+\rho\, H_{BC}H^\dagger_{BC}|<r\log\rho\big)}{\log\rho}\label{eq:dS_RD}\\
d_{MAC}(r)& \defeq -\lim_{\rho\rightarrow\infty}\frac{\underset{p(x_{i^*},x_r)}
{\min}\prob(I(x_{i^*}x_r;y_d)<
r\log\rho)}{\log\rho}\nonumber\\
&=-\lim_{\rho\rightarrow\infty}\frac{\prob\big(\log\big|I+2\rho\, H_{MAC}H^\dagger_{MAC}\;\big|<
r\log\rho\big)}{\log\rho},\label{eq:dSR_D}
\end{align}
The transmit signals $x_{i^*}$ and $x_r$ are from random codebooks that
are drawn according to complex Gaussian distributions with zero mean
and variance $\sqrt{\rho}$. We define
$H_{BC}\defeq\left[\begin{array}{c}h_{i^*r}\\h_{i^*d}\end{array}\right]$,
$H_{MAC}\defeq[h_{i^*d}\ h_{rd}]$ and $(\ )^\dagger$ denotes the
Hermitian operator. The derivation of
Equations~(\ref{eq:dS_RD}),~(\ref{eq:dSR_D}) uses the fact that a
constant scaling in the transmit power does not change the
DMT~\cite{Zheng2003}.

Using the techniques in~\cite{Zheng2003,Azarian2005} and following the
NAF MARC DMT proof, it is possible to calculate the following:
\begin{align}
d_{BC}(r)&=
\begin{cases} (n+1)-\frac{r}{t}, & r \le t< \frac{1}{n+1}\\
            n\frac{(1-r)}{(1-t)}, & t<\min(r,\frac{1}{n+1})\\
            (n+1)(1-r), &
            t\ge\frac{1}{n+1}\end{cases}\label{eq:DMTCF1}\\           
d_{MAC}(r)&=
\begin{cases}  (n+1)-\frac{r}{1-t}, & 1-r \ge t>\frac{n}{n+1} \\
            n\frac{(1-r)}{t}, & t>\max\{1-r,\frac{n}{n+1}\}\\
            (n+1)(1-r),&t\le\frac{n}{n+1}\end{cases}\label{eq:DMTCF2}
\end{align}
Details of the derivation are similar to, e.g.,
Theorem~\ref{theorem:SRC-DDF} and are omitted for brevity.

From Equations~(\ref{eq:DMTCF})-(\ref{eq:DMTCF2}), it follows that the
genie aided DMT upper bound can be achieved for any value of
$\frac{1}{n+1} \le t \le \frac{n}{n+1}$. The maximum achieved DMT is
given by
\begin{equation}
d(r)=(n+1)(1-r)^+
\label{eq:DMT-MARC-CF}
\end{equation}

\subsection{Optimality of the Achievable DMTs}
Although the previous DMTs were calculated using simplified selection
schemes that only observed the source-destination direct link, one can
show that for each of the relaying protocols, no improvement in DMT is
possible by more sophisticated selection schemes.

This fact is self-evident for the CF relaying result, since it meets
the genie-aided upper bound. The NAF and DDF do not meet the
genie-aided bound, therefore it is not obvious that they perform
optimally under the simplified selection scheme. We now proceed to
investigate this question for DDF and NAF.

The DMT of the multiple access relay channel with opportunistic user
selection is given by
\begin{align}
  d(r)=\lim_{\rho\rightarrow \infty}\frac{\log\prob({\cal
      O}_1,\ldots,{\cal O}_n)}{\log\rho},
\end{align}
where ${\cal O}_i$ represents the outage event for the access mode
characterized by source $i$ transmitting to the destination with the
help of the relay.

In a manner similar to~\cite{Azarian2005} and Equation~(\ref{eq:outgen}), the
probability of outage $\prob({\cal O}_1,\ldots,{\cal O}_n)$ can be
expressed as follows
\[
\prob({\cal O}_1,\ldots,{\cal O}_n)\doteq \rho^{-d_o(r)}, 
\]
where
\begin{align}
\label{eq:opt_MARC}
  d_o(r)=\underset{(v_1^{(1)},u^{(1)},\ldots,v_1^{(n)},u^{(n)},v_2)\in O
    }{\inf}v_2+\sum_{j=1}^{n}\big(v_1^{(j)}+u^{(j)}\big) \; .
\end{align}
The random variables $v_1^{(j)},u^{(j)}$  and $v_2$ represent the exponential order of
$1/|h_{jd}|^2,\; 1/|h_{jr}|^2$ and $1/|h_{rd}|^2$, respectively. Each
of these random variables has a probability density function that is
asymptotically equal to
\begin{align}
  p(x)\doteq
  \begin{cases}
    0& x<0\\ \rho^{-x} &x\ge 0.
  \end{cases}
\end{align}
The set $O$ represents the outage event for the opportunistic
network. We know $O={\cal O}_1^+\cap\ldots\cap{\cal O}_n^+$, i.e, the
opportunistic system is considered in outage when no access mode is
viable.

For NAF the outage region is defined by~\cite{Azarian2005}
\begin{align}
\label{eq:O_j_NAF}
  {\cal O}_j^+=\Big\{\big(v_1^{(j)},v_2,u^{(j)}\big)\in R^{3+}\Big|(l-2m)(1-v_1^{(j)})^++m \max\{2(1-v_1^{(j)}),1-(v_2+u^{(j)})\})^+<rl\Big\},
\end{align}
where m is rank of the relay amplification matrix and $l$ is the block
length. The solution to Equations~(\ref{eq:opt_MARC})
and~(\ref{eq:O_j_NAF}) is facilitated by the knowledge that $d_o(r)$
is maximized when $m=l/2$, leading to:
\begin{align}
  d_{NAF}(r)=n(1-r)+(1-2r),
\end{align}
This is the best diversity obtained for NAF, which is similar to the
simplified selection based on the source-destination link. Therefore
the optimality of the simplified selection rule is established for NAF.

For DDF the outage region is defined by~\cite{Azarian2005}
\begin{align}
\label{eq:O_j_DDF}
  {\cal O}_j^+=\Big\{\big(v_1^{(j)},v_2,u^{(j)}\big)\in R^{3+}\Big|t^{(j)}(1-v_1^{(j)})^++(1-t^{(j)}) (1-\min(v_1^{(j)},v_2))^+<r\Big\},
\end{align}
where $t^{(j)}$ is the listening-time ratio of the half-duplex relay
when source $j$ is transmitting, with $r\le t^{(j)}\le 1$. In the
following we outline the solution of
Equations~(\ref{eq:opt_MARC}) and~(\ref{eq:O_j_DDF}) for a two-user
MARC. The generalization to $n$ users is straight forward. 

Our strategy for solving the optimization problem is to partition the
optimization space into eight regions, solve the optimization problem
over each region as a function of $t^{(1)}$ and $t^{(2)}$, maximize
over $t^{(1)}$ and $t^{(2)}$ and then find the minimum of the eight
solutions. The eight regions correspond to the Cartesian product of
whether each of the three positive variables $v_1^{(1)},v_1^{(2)},
v_2$ is greater than or less than 1. 
Following the calculations, which are straight forward,
the DMT for DDF is
\begin{align}
  d_{DDF}(r)=
  \begin{cases}
    (n+1)(1-r)& \frac{n}{n+1}\ge r\ge 0\\
    n\frac{1-r}{r}& 1\ge r >\frac{n}{n+1},
  \end{cases}
\end{align}
which matches the DMT of simplified selection based on the
source-destination links. Therefore the optimality of simplified
selection for the DDF is established.

We can follow essentially the same steps for the broadcast relay
channel and obtain the same DMTs for both the NAF and DDF. The
optimization problem in the broadcast case is slightly different: the
shared link in BRC is the source-relay channel while it is the
relay-destination channel in the MARC. Nevertheless, very similar
strategies follow through for the BRC with only small adjustments.

%
\section{Opportunistic X-Relay Channel}
\label{sec:XRC}
The X-relay channel is defined as a $n\times n$ node wireless network
with a relay, where each of the $n$ sources has messages
for each of the $n$ destinations (see Figure~\ref{fig:XRC}). The
sources are not allowed to cooperate with each other, but the relay
cooperates with all sources.

\begin{figure}
\centering
\includegraphics{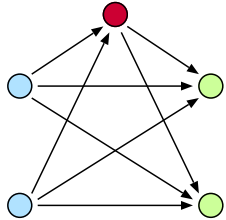}
\caption{The X-relay channel.}
\label{fig:XRC}
\end{figure}
There are only a few results available on the X channel, among them,
it has been shown~\cite{Jafar2008} that the X-channel with no relay
has exactly $\frac{4}{3}$ degrees of freedom when the channels vary
with time and frequency. The X-relay channel introduces a relay to the
X channel for improved performance.

The opportunistic X-relay channel has four access modes as shown in
Figure~\ref{fig:XRC1}. These modes avoid interference across different
message streams and satisfy our working definition of opportunistic
modes in relay networks.

\subsection{DMT upper bound}

To find an upper bound for the DMT of opportunistic X-relay channel,
we assume a genie transfers the data from the sources to the relay and
also allows the sources to know each other's messages. For the upper
bound we also allow the destinations to fully cooperate, noting that
it can only improve the performance.  Figure~\ref{fig:XRC2} shows the
genie-aided opportunistic modes, where the two-sided arrows indicate
the free exchange of information by the genie. From this figure, it is
easy to see that the genie-aided X-relay channel is equivalent to a
MIMO system with $3$ transmit antennas and $2$ receive antennas.

The performance of the {\em opportunistic} X-relay channel is therefore
upper bounded by a $3\times 2$ MIMO system with {\em antenna selection},
choosing for each codeword two transmitting and one receiving
antennas. It is noteworthy that the $3\times 2$ antenna selection allows
one configuration that does not have a counterpart in the opportunistic
modes in the X-relay channel, therefore due to the extra flexibility the
MIMO system with antenna selection upper bounds the performance of the
genie-aided opportunistic X-relay channel.

Using the result from Equation~(\ref{eq:Ant_sel}), a $3\times 2$ MIMO system
with two antennas selected from the transmitter side and one antenna
selected from the receiver side has a DMT that is upper bounded by
$d(r)=6(1-r)^+$. This in turn is an upper bound to the performance of
the opportunistic X-relay channel.

\begin{figure*}
\begin{minipage}[b]{3.0in}
\centering 
\includegraphics[height=2in]{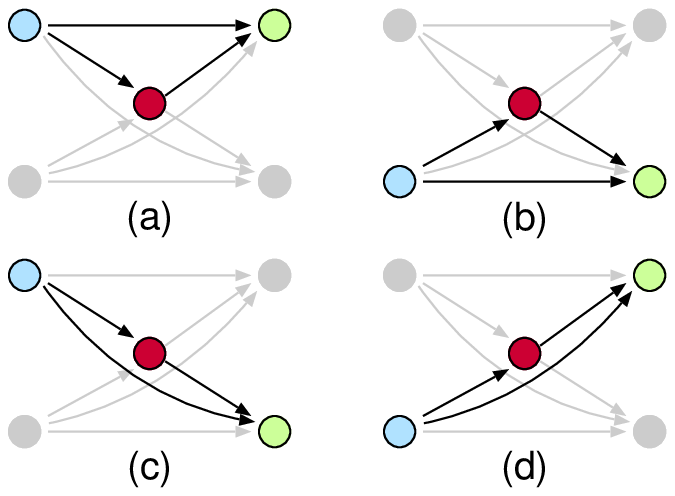}
\caption{Opportunistic modes of the X-relay channel}
\label{fig:XRC1}
\end{minipage}
\hfill
\begin{minipage}[b]{3.0in}
\centering 
\includegraphics[height=2in]{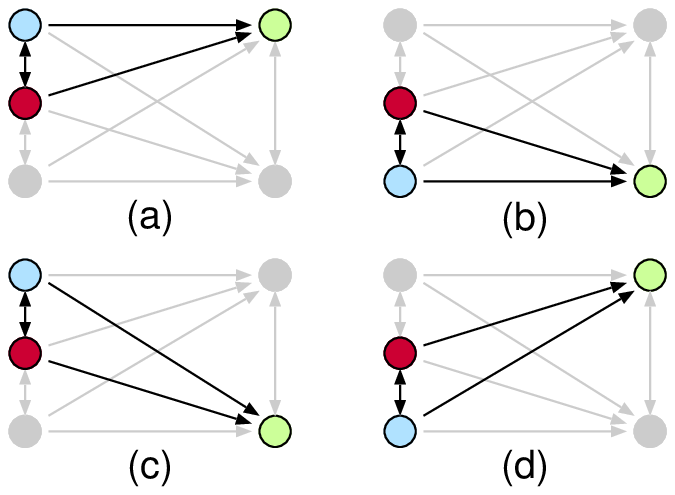}
\caption{Opportunistic modes of  genie-aided X-relay channel.}
\label{fig:XRC2}
\end{minipage}
\end{figure*}

\subsection{Achievable DMT}

For deriving achievable rates, we consider the following simplified
opportunistic scheme. First, we choose between the two access modes
(a) and (b) in Figure~\ref{fig:XRC1}. If both these two modes are in
outage, we consider {\em only} the direct link of the two access modes
(c) and (d), i.e., the relay is not allowed to cooperate in modes (c)
and (d). Note that this is only a simplification for the purposes of
achievable-DMT analysis, the idea being that if the relay is useful in
neither of the access modes (a) and (b), it is unlikely to be useful
at all. The approximation involving the conditional removal of the
relay from (c) and (d) allows the access modes to become independent
and simplifies the analysis. The resulting achievable rate is tight
against the upper bound for compress-forward, as seen in the sequel,
but not demonstrably so for other protocols.

Access modes (a) and (b) do not share any common links, therefore
their statistics are independent. Each of them is an ordinary relay
channel which can achieve $d(r)=2(1-r)^+$ via the CF
protocol~\cite{Yuksel2007}. The (c) and (d) access modes, which were
reduced to a single link, each achieves the DMT
$d(r)=(1-r)^+$. Furthermore, the source-destination links in (c) and
(d) are disjoint from the links in (a) and (b), therefore the
statistics are independent and we can use Lemma~\ref{lemma:1} to find
the overall DMT $d(r)=6(1-r)^+$. Note that this achievable DMT meets
the upper bound, therefore the DMT of the X-relay channel under CF is
exactly $d(r)=6(1-r)^+$.

Achievability results for relaying protocols other than CF can be
obtained along the same lines. We begin with NAF. Recall that the DMT
of a simple relay network (source, relay, destination) under NAF is
$d(r) = (1-r)^+ + (1-2r)^+$. Combining the four access modes (a),
(b), (c), (d) mentioned earlier for the X-relay channel together with
the NAF protocol results in:
\begin{align}
d_{XNAF}(r) &= 2(1-r)^+ + 2 \Big [ (1-r)^+ + (1-2r)^+\Big] \nonumber\\
&= 4(1-r)^+ + 2(1-2r)^+
\end{align}
A similar result exists for the DDF where the DMT is given by
\begin{align}
d_{XDDF}(r)=
\begin{cases}
  6(1-r)&0\le r<\frac{1}{2}\\
2\frac{1-r}{r}+2(1-r)& \frac{1}{2}\le r \le1
\end{cases}
\end{align}
Applying the same analysis to orthogonal AF and DF yields a diversity
$d(r)=2(1-r)^+ +4(1-2r)^+$, but there is more to be said for
orthogonal transmission. In orthogonal transmission it may be
beneficial at high multiplexing gains to shut down the relay, therefore a
complete analysis requires two more opportunistic modes that are
derived by shutting down the relay from modes (a) and (b).  Using this
extended set of six access modes, the DMT of the opportunistic X-relay
channel with orthogonal AF or orthogonal DF is
\begin{equation}
d(r)=4(1-r)^++2(1-2r)^+
\end{equation}
which matches the DMT of NAF. 

Thus far, to find achievable DMTs for the X relay channel we used
simplified selection rules and access modes. In the case of CF, this
simplified achievable DMT is in fact optimal since it matches the
genie upper bound. Other protocols do not meet the genie-aided bound,
therefore the question of the optimality of simplified selection for
other protocols is more involved. Nevertheless, for NAF and DDF also,
no DMT gains can be obtained by more sophisticated selection rules and
access modes, as outlined below.

To find the overall optimal DMT without the simplifications, we need
to solve a linear optimization problem similar to~(\ref{eq:opt_MARC})
where
\begin{align}
  d_o(r)=\underset{(v_1^{(ij)},v_2^{(rj)},u^{(jr)})\in O,\; i,j\in\{1,2\}
    }{\inf}\sum_{j=1}^{2}\Big(\sum_{i=1}^{2}v_1^{(ij)}+v_2^{(rj)}+u^{(jr)}\Big),
\end{align}
where $v_1^{(ij)},v_2^{(rj)}$ and $u^{(jr)}$ represent the exponential
order of $1/|h_{ij}|^2,\; 1/|h_{rj}|^2$ and $1/|h_{jr}|^2$,
respectively. The outage event $O$ is characterized by ${\cal
  O}_1^+\cap{\cal O}_2^+\cap{\cal O}_3^+\cap{\cal O}_4^+$, i.e., the
system is in outage if all access modes are in outage. The outage
event is given by Equation~(\ref{eq:O_j_NAF}) for NAF and
Equation~(\ref{eq:O_j_DDF}) for DDF. In a straight forward manner, the
optimization above gives the same DMTs found by the simplified
selection criterion, therefore the calculated DMTs cannot be improved
upon and are optimal.

%
%

\begin{figure*}
\centering
\includegraphics[width=5in]{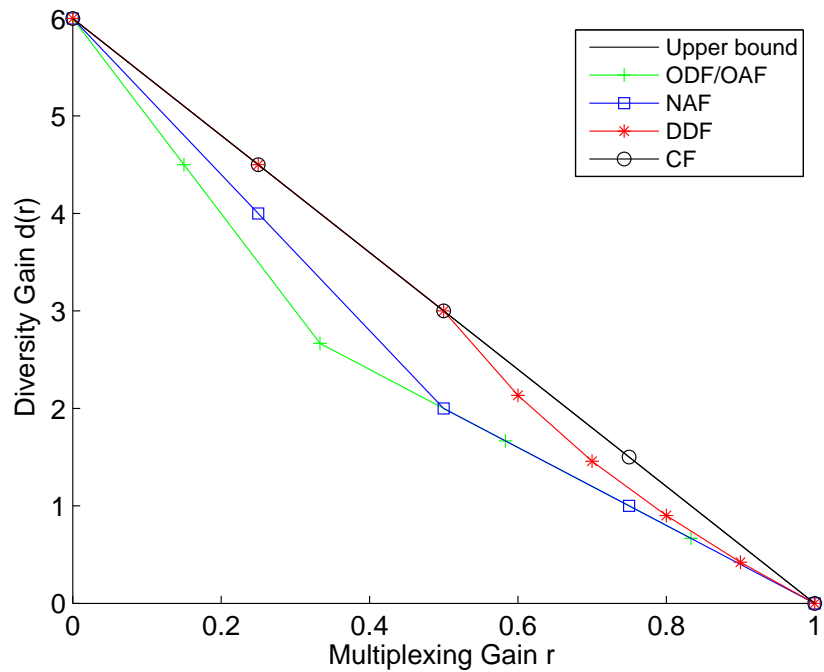}
\caption{The DMT of the opportunistic X-relay channel under orthogonal
  AF and DF, Non-orthogonal AF, Dynamic DF, and CF. The CF achieves the
  DMT upper bound.}
\label{fig:XRCDMT}
\end{figure*}

\section{The Gateway Channel}
\label{sec:GWC}
The gateway channel~\cite{Abouelseoud2008} is a multi-node network with
$M$ source-destination pairs that communicate with the help of a relay
(see Figure~\ref{fig:GWC}). Each source communicates only with its
corresponding destination. A two-hop communication scheme is used, where
at the first hop the sources transmit to the relay and at the second hop
the relay transmits to the destinations. No direct link exists between
the sources and destinations, therefore if the relay is in outage the
destination will surely be in outage. Under these conditions, the most
natural mode of operation is decode-and-forward, although
amplify-and-forward may also be considered due to practical
limitations. In this work we concentrate on the DF gateway channel.


We do not require data buffering at the relay. With an infinite buffer
at the relay, the gateway channel decomposes into a concatenation of a
MAC and a broadcast channel. An infinite buffer would thus simplify the
analysis but also increase the overall latency and relay complexity. One
of the interesting outcomes of the forthcoming analysis is that that
data buffering in the asymptotic high-SNR regime does not provide a
performance advantage (in the sense of DMT).

We start with the non-opportunistic gateway channel, and then move to
the opportunistic scheme.

\subsection{No Transmit CSI}

We first consider the case where all nodes have receive-side CSI, but
the nodes, and in particular the relay, do not have transmit-side
CSI. Under these conditions, we cannot choose source-destination pairs
according to their SNR. Then the choice of transmission strategies on
the MAC and broadcast side of the network are as follows.

On the broadcast side, the channel gains are random and unknown to the
relay. In light of symmetric rate requirements, the transmit strategy
must be symmetric with respect to the destinations. Under this
symmetry, the best achievable rate is according to orthogonal
transmission~\cite{tse2005} and superposition coding does not give
better results.

For the multiple-access side, under symmetric rate requirement, both
orthogonal and superposition channel access are viable. It has been
shown that superposition access gives slightly better performance at
medium SNR, while at high and low SNR the two methods have
asymptotically the same capacity under symmetric
rates~\cite[pp. 243-245]{tse2005}.

\begin{figure}
\centering \includegraphics{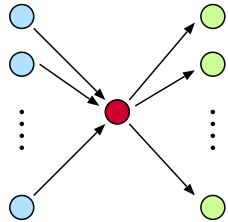}
\caption{The Gateway channel.}
\label{fig:GWC}
\end{figure}

In the absence of transmit-side CSI, and with symmetric rate
requirements, the network does indeed decompose into a cascade of a
multiple-access and a broadcast subnetworks, and the overall outage
probability is given by:
\begin{align}
P_{\cal O}&= 1- (1-P_{MAC})(1-P_{BC})\nonumber\\
&=P_{MAC}+P_{BC}-P_{MAC}P_{BC}.
\end{align}
Where $P_{MAC}$ (respectively $P_{BC}$) denotes the outage of the MAC
(respectively broadcast channel), defined as the probability that one
or more of the users in the MAC (respectively broadcast channel)
cannot support rate $R$. In a slowly fading environment, for a power
allocation vector $P_s=(P_1,\ldots,P_M)$, a fading state
$H=(h_{1r},\ldots,h_{Mr})$ and superposition coding, the outage is
given by $P_{MAC}=\prob \big(\bar R\notin {\cal C}_{MAC}\big)$ where
\begin{align}
{\cal C}_{MAC}(H,P)=\bigg \{ \bar R:\sum_{i\in S}R_i\le
\frac{1}{2}\log\bigg(1+\frac{1}{N}\sum_{i\in S}|\hir |^2P_i\bigg)\bigg\},
\end{align}
With a time-sharing MAC, the outage probability is
\begin{align}
P_{MAC}&=
\prob \bigg (\Big\{R_i>\frac{1}{2M} \log\big(1+\rho|\hir |^2\big), \;  i=1,\ldots,M\Big\}\bigg)
\end{align}

On the broadcast side, the outage is given by $P_{BC}=\prob \big(\bar R\notin
{\cal C}_{BC}\big)$, where \begin{equation} {\cal
    C}_{BC}=\bigg\{\bar R: R_i \le \frac{1}{2M}
  \log\big(1+\rho|h_{ri}|^2\big)\bigg\},
\end{equation}

Without transmit CSI, the DMT is the minimum of the DMT of the MAC and
the broadcast channel.  For the MAC channel, it has been
shown~\cite{Tse2004} that for multiplexing gains $r \le\frac{M}{M+1}$,
the diversity $d=1-r/M$ is achievable, while for higher rates
$\frac{M}{M+1} < r \le 1$, the diversity of $d=M(1-r)$ is obtained.

For the broadcast channel, since time sharing achieves the maximum
sum-rate bound, the broadcast DMT is similar to the single-user
DMT. The DMT of the network is bounded by the DMT of the broadcast
part of the network. Thus, including the half-duplex consideration,
the best achievable DMT is
\begin{equation}
d(r)=(1-2r)^+.
\end{equation}
The same DMT can be obtained with orthogonal channel access;
superposition coding has no effect on the DMT.

\subsection{Opportunistic Channel Access}

In this scenario, the relay is assumed to have channel state information
(either perfect or incomplete) about its incoming and outgoing
links. Using this information, during each transmission interval the
relay selects the best overall source-destination pair and gives it
access to the channel. Form Lemma~\ref{lemma:1}, it is easy to see that
the DMT of an opportunistic gateway channel is upper bounded by $d(r)\le
n(1-2r)^+$. We start by assuming perfect CSI at the relay.

\subsubsection{Full CSI at the Relay}

We start by defining
\[
\gamma_i \defeq  \min(|h_{ir}|^2,|h_{ri}|^2) \; .
\]
In the decode-and-forward protocol, end-to-end data transmission is
feasible if and only if both source-relay and relay-destination links
can support the desired rate, therefore $\gamma_i$ is the effective
channel gain that governs the rate supported by a DF protocol for any
node pair $i$. In the opportunistic mode, we would like to support the
maximum instantaneous rate, therefore user $i^*$ will be selected such
that:
\begin{equation}
i^*=\arg\max_i  \gamma_i,
\end{equation}
We now investigate the statistics of $\gamma_{i^*}$. Since the channel
fading coefficients $\hri $ and $\hir $ are complex Gaussian random
variables, the channel gains $|\hri |^2$ and $|\hir |^2$ obey
exponential distributions with exponential parameters $\frac{1}{E[|\hri
    |^2]}$ and $\frac{1}{E[|\hir |^2]}$, respectively. It is known that
the minimum of $M$ exponential random variables with parameters
$\lambda_k$ is an exponential random variable with parameter
$\sum_{k=1}^{M}\lambda_k$, therefore the pdf of $\gamma_i$ is an
exponential distribution with parameter $\lambda=2$. Therefore, the cdf
of the maximum SNR for all the source-relay-destinations links
$\gamma_{i^*}$ is
\begin{align}
F_{\gamma_{i^\ast}}(x)&=\left(1-e^{-2 x}\right)^{M}
\label{eq:gammaCDF}
\end{align}
The network is considered in outage when none of the source-destination
pairs can support the desired transmission rate $R$. The outage
condition is therefore:
\begin{align}
P_{\cal O}&=\prob 
\Big(R> \frac{1}{2}\log\big(1+\rho\gamma_{i^*})\Big)\nonumber\\ 
&=\prob \Big(\gamma_{i^*} < \frac{\rho^{2r}-1}{\rho}\Big) \nonumber\\  
&=\bigg(1-\exp\Big(-2\frac{\rho^{2r}-1}{\rho}\Big)\bigg)^{M}.
\end{align}
The block sizes in our analyses are large enough so that the error
events are dominated by outage events, therefore the probability of
error can be approximated by the outage probability. Using the Taylor
approximation $1-\exp(-x) \approx x$, we get:
\begin{align}
P_e&\doteq \Big ( \frac{\rho^{2r}-1}{\rho}\Big)^M\nonumber \\
&\doteq \rho^{-M(1-2r)} \; .
\end{align}
where the Taylor approximation is valid for $2r<1$. Hence, the
opportunistic gateway channel achieves the following DMT
\begin{equation}
d(r)=M(1-2r)^+.
\end{equation}

\begin{remark}
If the path selection criterion uses one set of channel gains,
i.e. either $\{\hir\}$ alone or $\{\hri\}$ alone, no diversity gain
would result. For example, selecting on the MAC side of the network
would give $\gamma_i=\min(|h_{i^*r}|^2,|h_{ri^*} |^2)$ where $i^*=\arg
\max|\hir|^2$. Since the channel gains on the two sides are independent,
$|h_{ri^*}|^2$ is still exponential and dominates the diversity order.
\end{remark}

\begin{remark}
The outage calculations assume that upon selection each source must be
connected to its corresponding destination within one transmission
interval, implying that no long-term storage and buffering is taking
place at the relay. In addition to simplifying the relay, this is also
helpful in terms of reducing the end-to-end delay due to opportunistic
communication.
\end{remark}

\begin{remark}
An infinite buffer at the relay may increase the throughput, but it does
not improve the DMT. If the relay can hold onto the data, the incoming
packets could wait indefinitely until the path to their destination is
dominant. Under this condition, the opportunistic MAC and opportunistic
broadcast operations can be performed independently, each giving rise to
a diversity $d=M(1-2r)^+$, thus the overall diversity would also be
$d=M(1-2r)^+$. However, this is no more than the diversity obtained
without the buffer.
\end{remark}

To summarize, a buffer would not improve the DMT, however, it would
allow us to achieve the optimal DMT via local decision making (using
MAC information on the MAC side, and broadcast channel information on
the broadcast side). Without buffering, the relay must make decisions
jointly in order to achieve optimal DMT.

\subsubsection{Limited Feedback}

We now assume the relay does not have perfect CSI but rather has
access to one bit of information per node from each destination and
is further able to send one bit of information per node to each of the
sources. We wish to explore the DMT of this network under
the one-bit feedback strategy.

Each destination node knows its incoming channel gain via the usual
channel estimation techniques. Each destination compares its incoming
channel gain to a threshold $\alpha$, reporting the result via the
one-bit feedback to the relay. The $k$ destination nodes that report
``1''(and their respective channels) are characterized as eligible for
data transmission in that interval. From among these $k$ eligible
destinations, the relay chooses the one whose corresponding
source-relay channel is the best.

The network is considered in outage if there is no
source-relay-destination link that can support the target rate $R$. We
design the threshold of the second hop of the network such that each
destination reports ``1'' if the corresponding relay-destination link
can support this rate $R$, i.e., $\alpha=\frac{\rho^{2r}-1}{\rho}$. The
outage event occurs if no destination reports positively, or if some
destinations are eligible, but none of the corresponding source-relay
links can support the rate $R$.  If according to this methodology the
relay detects more than one end-to-end path that can support the rate
$R$, the relay selects one of them randomly.

We define $A_m$ as the event of $m$ destinations reporting ``1'', and
$\prob(e|A_m)$ as the probability of error given that $m$ destinations
report ``1''. This is the probability that none of the $m$
eligible relay-destination channels have a corresponding source-relay
link that can support the rate $R$. The probability of outage in this case is
\begin{align}
\label{eq:out_lmt}
P_{\cal O}&=\prob(A_0)+\sum_{m=1}^M \prob(A_m) \prob(e|A_m).
\end{align}
The probability of $m$ destinations reporting ``1'' and $M-i$
destinations reporting ``0'' is
\begin{align}
\label{eq:Prb_ai}
\prob(A_m)&={M \choose m} F_{\gamma}(\alpha)^{m}\big(1-F_\gamma(\alpha)\big)^{M-m}\nonumber\\
&={M \choose m}\big(e^{-\lambda \alpha}\big)^m\big(1-e^{-\lambda \alpha}\big)^{M-m},
\end{align}
where $F_\gamma(x)$ is the cdf of the channel gains $\gamma=|h|^2$,
which is exponentially distributed with parameter $\lambda=1$. The
probability of error given that $m$ destinations report ``1'' is
\begin{align}
\label{eq:prb_eror}
\prob(e|A_m)&=\prob\big(\max_{j\in S} |h_{jr}|^2 \le
\alpha\big)\nonumber\\
&=\big(1-F_\gamma(\alpha)\big)^m=\big(1-e^{-\lambda \alpha}\big)^{m},
\end{align}
where $S\subset\{1,\ldots,M\}$, $|S|=m$, and we use the fact that
source-relay and relay-destination channel gains have the same
distribution $F_\gamma$. Assuming non-identical exponential
distributions introduces more variables into analysis but the end
results will be
identical. Substituting~(\ref{eq:Prb_ai}),~(\ref{eq:prb_eror})
in~(\ref{eq:out_lmt}), the outage probability becomes
\begin{align}
P_{\cal O}&=\big(1-e^{-\lambda \alpha}\big)^{M}
+\sum_{m=1}^M {M \choose m}\big(e^{-\lambda \alpha}\big)^m\big(1-e^{-\lambda \alpha}\big)^{M-m}\big(1-e^{-\lambda \alpha}\big)^{m}\nonumber\\
&=\sum_{m=0}^M {M \choose m}\big(e^{-\lambda_g \alpha}\big)^m\big(1-e^{-\lambda \alpha}\big)^{M-m}\big(1-e^{-\lambda \alpha}\big)^{m}.
\end{align}

To calculate the DMT, from~(\ref{eq:out_lmt}), the outage probability is
\begin{align}
P_{\cal O}&=\prob \bigg(\frac{1}{2}\log\big(1+\max_i|\hri |^2\rho\big)\le
r\log\rho\bigg)
+\sum_{m=1}^M {M\choose m}\prob \bigg(\frac{1}{2}\log
\big(1+|h_{rd}|^2\big)\le r\log\rho\bigg)^{M-m}
\nonumber\\& \qquad \times 
\prob \bigg(\frac{1}{2}\log\big(1+|h_{rd}|^2\big)\ge
r\log\rho\bigg)^{m}
\prob \bigg(\frac{1}{2}\log \big(1+\max_{j\in
S,|S|=m}|h_{jr}|^2\big)\le r\log\rho\bigg)\nonumber\\
&\doteq \prob \bigg( \max_i|\hri |^2 \le
\rho^{2r-1}\bigg)
+\sum_{m=1}^M \prob \big(|h_{rd}|^2 \le \rho^{2r-1} \big)^{M-m}
\prob \big(|h_{rd}|^2\ge \rho^{2r-1}\big)^m \nonumber\\
&\qquad \qquad \times \prob \bigg(\max_{j\in S, |S|=m} |h_{jr}|^2 \le
\rho^{2r-1}\bigg)\nonumber\\
&\doteq \sum_{m=1}^M \big(e^{-\lambda\rho^{2r-1}}\big)^m\big(1-e^{-\lambda
\rho^{2r-1}}\big)^{M-m}\big(1-e^{-\lambda \rho^{2r-1}}\big)^{m}\nonumber\\
&\doteq \rho^{M(2r-1)}.
\end{align}
So we have:
\begin{equation}
d(r)=M(1-2r)^+.
\end{equation}
Thus, even 1-bit feedback is enough to achieve optimal
DMT.
\section{conclusion}
\label{sec:conclusion}

The high-SNR performance of opportunistic relay networks are
investigated. Except for a handful of simple relay selection scenarios,
there are two main difficulties in the analysis of opportunistic relay
networks: (1) the decision variables often depend on more than one link
gain, complicating the performance analysis and (2) the opportunistic
modes may share links and thus are statistically dependent, which
complicates the order statistics that govern the performance of
opportunistic systems. In this work, several relaying geometries are
studied and the corresponding DMTs are developed for a number of
well-known relaying protocols, including the AF, DF, CF, NAF, and
DDF. In several instances, selection schemes based on the direct
source-destination links are shown to achieve optimal performance, for
example the CF multiple access channel. In some network geometries,
opportunistic selection using 1-bit feedback is shown to achieve the
optimal DMT performance.
Future work may include investigating variations in the channel state
knowledge in the performance of the system, for example the effect of
partial CSIT at the nodes~\cite{Kim2007}.

\appendices

\section{Opportunistic DF orthogonal relaying over a simple relay channel}
\label{Appen:ODF}
 The DMT of the opportunistic orthogonal relaying is given by
\begin{equation}
d(r)=d_1(r)+d_2(r),
\label{eq:DMT_O}
\end{equation}
where 
\begin{align}
d_1(r)&=\lim_{\rho\rightarrow\infty}\frac{\log\prob(e_1)}{\log\rho},
\label{eq:DMT_O1}\\
d_2(r)&=\lim_{\rho\rightarrow\infty}\frac{\log\prob(e_2|e_1)}{\log\rho}.
\label{eq:DMT_O2}
\end{align}
The events $e_1$ and $e_2$ represent the error in the non-relay and
the relay-assisted modes, respectively. The non-relay access mode is a
simple direct link, whose DMT is $d_1(r)=(1-r)^+$. The DMT of the
relay-assisted access mode is known, however, the DMT of the relay
channel {\em conditioned on} the outage event of the direct link
requires new calculations.

Recall that the orthogonal DF relaying works as follows: The
transmission interval is divided into two halves. In the first half, the
source transmits. If the relay cannot decode the source message, it will
remain silent and the source will continue to transmit into the second
half-interval. If the relay decodes the source message, the relay
forwards the decoded message to the destination in the second half of
the transmission interval and the source remains silent.
  
Because of orthogonality and with the use of long codewords, it is
trivial to see that error is dominated by outage. The conditional
outage probability of the relay-assisted mode is given by
\begin{align}
\prob({\cal O}_2|{\cal O}_1)&=\prob\bigg(\Big\{\frac{1}{2}\log(1+U\rho)<r\log \rho\Big\}\Big|\Big\{\log(1+|h_{sd}|^2\rho)<r\log\rho\Big\}\bigg)\\
&=\prob\bigg(\Big\{U<\frac{\rho^{2r}-1}{\rho}\Big\}\Big|\Big\{|h_{sd}|^2<\frac{\rho^{r}-1}{\rho}\Big\}\bigg),
\label{eq:out_ODF}
\end{align}
where the random variable $U$ is given by
\begin{equation}
  U=
  \begin{cases}
    2|h_{sd}|^2 & |h_{sr} |^2 <\frac{\rho^{2r}-1}{\rho}\\ 
     |h_{sd}|^2+|h_{rd} |^2 & |h_{sr} |^2\ge \frac{\rho^{2r}-1}{\rho}.
  \end{cases}
\end{equation}
The cdf of $U$ is given by
\begin{align}
F_U(u)&=\prob\Big(|h_{sd} |^2<\frac{u}{2}\Big)\prob\bigg(|h_{sr} |^2 <
\frac{\rho^{2r}-1}{\rho}\bigg)
+\prob\Big(|h_{sd} |^2+|h_{rd}|^2<u\Big)\prob\bigg(|h_{sr} |^2
\ge \frac{\rho^{2r}-1}{\rho}\bigg).\nonumber
\end{align}
Hence,
\begin{align}
\prob({\cal O}_2|{\cal O}_1)=&
\prob\bigg(\Big\{|h_{sd} |^2<\frac{1}{2}\,\frac{\rho^{2r}-1}{\rho}\Big\}\Big|\Big\{|h_{sd}|^2<\frac{\rho^{r}-1}{\rho}\Big\}\bigg)\prob\bigg(|h_{sr} |^2 <\frac{\rho^{2r}-1}{\rho}\bigg)\nonumber\\
&+\prob\bigg(\Big\{|h_{sd}
|^2+|h_{rd}|^2<\frac{\rho^{2r}-1}{\rho}\Big\}\Big|\Big\{|h_{sd}|^2<\frac{\rho^{r}-1}{\rho}\Big\}\bigg)\prob\bigg(|h_{sr}
|^2\ge \frac{\rho^{2r}-1}{\rho}\bigg)
\label{eq:ODF-conditional}
\end{align}
One can show that $\frac{1}{2}\,\frac{\rho^{2r}-1}{\rho}\dot > \frac{\rho^{r}-1}{\rho}$, therefore 
\begin{equation} 
\prob\Big(\Big\{|h_{sd} |^2<\frac{1}{2}\,\frac{\rho^{2r}-1}{\rho}\Big\}\Big|\Big\{|h_{sd}|^2<\frac{\rho^{r}-1}{\rho}\Big\}\Big)\doteq 1.
\label{eq:cond_exp}
\end{equation}

To analyze the second conditional term in
Equation~(\ref{eq:ODF-conditional}), we begin with the pdf of
$Z=|h_{sd} |^2+|h_{rd} |^2$ conditioned on the event
$B=\Big\{|h_{sd}|^2<\frac{\rho^{r}-1}{\rho}\Big\}$. The channel gain
$\gamma \defeq |h_{sd} |^2$ has the following conditional distribution
\begin{equation}
\label{eq:exp_cond2}
  f_{\gamma|B}(x)=\begin{cases} \frac{
    e^{-x}}{1-e^{- \frac{\rho^r-1}{\rho}}} &
  x\le\frac{\rho^{r}-1}{\rho}, \\ 0&x
  >\frac{\rho^{r}-1}{\rho}. \end{cases}
\end{equation}
Defining $g_1(r,\rho)\defeq\frac{\rho^r-1}{\rho}$ and
$g_2(r,\rho)\defeq\frac{\rho^{2r}-1}{\rho}$, the conditional pdf of
$Z=|h_{sd} |^2+|h_{rd} |^2$ is calculated as follows, for $z\le
g_1(r,\rho)$
\begin{align}
  f_{Z|B}(z)&=\int_0^z e^{- (z-x)}\frac{ e^{-
      x}}{1-e^{- g_1(r,\rho)}}dx\nonumber\\ &=\frac{z
    e^{- z}}{1-e^{- g_1(r,\rho)}}.
\end{align}
For $z>g_1(r,\rho)$, the conditional pdf of $Z=|h_{sd} |^2+|h_{rd} |^2$ is
given by
\begin{align}
  f_{Z|B}(z)&=\int_0^{g_1(r,\rho)}e^{- (z-x)}\frac{
    e^{-x}}{1-e^{-
      g_1(r,\rho)}}dx\nonumber\\ &=\frac{ g_1(r,\rho)
    e^{- z}}{1-e^{- g_1(r,\rho)}}.
\end{align}

The conditional probability of outage is
calculated as follows
\begin{align}
\prob\Big(\Big\{|h_{sd} |^2+|h_{rd}|^2<g_2(r,\rho)\Big\}\Big|&\Big\{|h_{sd}|^2<g_1(r,\rho)\Big\}\Big)\nonumber\\
&=\int_{0}^{g_1(r,\rho)} \frac{z
    e^{-z}}{1-e^{-g_1(r,\rho)}}dz+\int_{g_1(r,\rho)}^{g_2(r,\rho)}  \frac{g_1(r,\rho) e^{-z}}{1-e^{-g_1(r,\rho)}}dz\nonumber\\
&= \frac{1- e^{-g_1(r,\rho)}
-g_1(r,\rho)e^{-g_2(r,\rho)}}{1- e^{-g_1(r,\rho)}}\nonumber\\
&\doteq  1- \frac{\rho^{r-1} e^{-\rho^{2r-1}}} {1- e^{-\rho^{r-1}}}\doteq  \rho^{2r-1}.\label{eq:cond_gamma}
\end{align}
Substituting (\ref{eq:cond_exp}) and (\ref{eq:cond_gamma}) into (\ref{eq:out_ODF}), the conditional probability of outage is given by
\begin{align}
P({\cal O}_2|{\cal O}_1)\doteq&\rho^{(2r-1)}+\rho^{(2r-1)}(1-\rho^{(2r-1)})\nonumber\\
\doteq&\rho^{(2r-1)}.
\label{eq:out_cond_fn}
\end{align}
Using Equations (\ref{eq:DMT_O}), (\ref{eq:DMT_O1}), (\ref{eq:DMT_O2})
and (\ref{eq:out_cond_fn}), the DMT of the orthogonal opportunistic
DF relaying is given by
\begin{equation}
d(r)=(1-r)^++(1-2r)^+.
\end{equation} 

\section{Opportunistic AF orthogonal relaying over a Simple Relay Channel}
\label{Appen:OAF}

The outage probability of the relay-assisted mode, given that the
non-relay mode is in outage is given by
\begin{align}
  \prob({\cal O}_2|{\cal O}_1)&=\prob\bigg(\Big\{\frac{1}{2} \log \Big(1+|h_{sd}|^2\rho+ f\big(|h_{sr}|^2\rho,|h_{rd}|^2\rho\big)\Big)<r\log\rho\Big\}\Big|\Big\{\log(1+|h_{sd}|^2\rho)<r\log\rho\Big\}\bigg)\\
&=\prob\bigg(\Big\{|h_{sd}|^2+ \frac{1}{\rho}f\big(|h_{sr}|^2\rho,|h_{rd}|^2\rho\big)<\frac{\rho^{2r}-1}{\rho}\Big\}\Big|\Big\{|h_{sd}|^2<\frac{\rho^{r}-1}{\rho}\Big\}\bigg),\label{eq:out_con_AF}
\end{align}
At high SNR, Equation~(\ref{eq:out_con_AF}) can be approximated by
\begin{align}
  \prob({\cal O}_2|{\cal O}_1)&=\prob\bigg(\Big\{|h_{sd}|^2+ \frac{|h_{sr}|^2|h_{rd}|^2}{|h_{sr}|^2+|h_{rd}|^2}<\frac{\rho^{2r}-1}{\rho}\Big\}\Big|\Big\{|h_{sd}|^2<\frac{\rho^{r}-1}{\rho}\Big\}\bigg)
\end{align}
where $\frac{|h_{sr} |^2|h_{rd} |^2}{|h_{sr} |^2+|h_{rs}|^2}$
represents the harmonic mean of two independent exponential random
variables. Using the result of~\cite{Seddik2006}, the harmonic mean of
two exponential random variables with exponential parameters $\lambda$
can be approximated by an exponential random variable with exponential
parameter $2\lambda$.

In order to calculate the conditional outage probability distribution,
we first calculate the conditional density function of
$Z=|h_{sd}|^2+V$ where $V=\frac{|h_{sr} |^2|h_{rd} |^2}{|h_{sr}
  |^2+|h_{rs}|^2}$. Again, we are assuming
$g_1(r,\rho)=\frac{\rho^{r}-1}{\rho}$, $g_2(r,\rho)=\frac{\rho^{2r}-1}
{\rho}$, and conditioning is over the event $B=\Big\{h_{sd}
|<\frac{\rho^{r}-1}{\rho}\Big\}$. The conditional probability density
function of $Z=|h_{sd} |^2+V$ is given by
\begin{equation}
   f_{Z|B}(z)=
   \begin{cases}
     \frac{2e^{- 2z}(e^{z}-1)}{1-e^{- g_1(r,\rho)}}&z\le g_1(r,\rho)\\
     \frac{2e^{-2z}\big(e^{g_1(r,\rho)}-1\big)}{1-e^{- g_1(r,\rho)}}&z>g_1(r,\rho).
   \end{cases}
\end{equation}
The conditional probability of outage is calculated as follows
\begin{align}
\prob\Big(|h_{sd} |^2+|h_{rd}|^2<g_2(r,\rho)\Big|&|h_{sd}|^2<g_1(r,\rho)\Big)\nonumber\\
&=2\int_{0}^{g_1(r,\rho)} \frac{e^{- 2z}(e^{z}-1)}{1-e^{- g_1(r,\rho)}}dz+2\int_{g_1(r,\rho)}^{g_2(r,\rho)}  \frac{e^{-2z}\big(e^{g_1(r,\rho)}-1\big)}{1-e^{- g_1(r,\rho)}}dz\nonumber\\
&=\frac{e^{-2g_2(r,\rho)}-e^{-g_1(r,\rho)}-e^{-2g_2(r,\rho)+g_1(r,\rho)}+1}{1-e^{-g_1(r,\rho)}}\nonumber\\
&\doteq1+e^{-2\rho^{2r-1}}\frac{1-e^{\rho^{r-1}}}{1-e^{-\rho^{r-1}}}\doteq \rho^{2r-1}.
\label{eq:out_cond_fn2}
\end{align}
Using Equations (\ref{eq:DMT_O}), (\ref{eq:DMT_O1}), (\ref{eq:DMT_O2}),
and (\ref{eq:out_cond_fn2}), the DMT of the orthogonal opportunistic
AF relaying is given by
\begin{equation}
d(r)=(1-r)^++(1-2r)^+.
\end{equation}

\section{Genie-Aided DMT Upper Bound for the Shared Relay Channel}
\label{Appen:SRC2}

The indexing of the access modes does not affect the problem, therefore
we can order the conditional events in
Lemma~\ref{lemma:1},~\ref{lemma:2} arbitrarily. In the following, we
index the outage events according to the order of selection that is
described below, which is designed to sort out the dependencies in a way
to make computations tractable.

The selection algorithm is as follows: If the non-relayed configuration
(shown in Figure~\ref{fig:SRC3} part (c)) can support the required rate
$R=r\log\rho$, it is selected. We shall call this {\em Mode 1} in the
remainder of appendices. If {\em Mode 1} is in outage (an event denoted
by ${\cal U}_1$) we will check to see if either of the two direct links
can individually support half the rate, i.e.,
$R=\frac{r}{2}\log\rho$. If one of the direct links can support this
reduced rate, we consider the relayed mode sharing that direct link. (If
none of the direct links can even support half the rate, we can consider
either one at random.) This relayed mode shall be called {\em Mode 2}.
If {\em Mode 2} can support the full required rate, it is selected. The
outage of {\em Mode 2} is denoted ${\cal U}_2$. If both {\em Modes 1, 2}
are in outage, the remaining relayed mode, which will be denoted {\em
  Mode 3}, is selected. The outage of {\em Mode 3} is denoted ${\cal
  U}_3$ in this and the following appendices. The error events
corresponding to the three modes are denoted $e_1', e_2', e_3'$ in this
and subsequent appendices. 

The total DMT of the genie-aided system is
\begin{align}
\label{eq:dmtadd}
  d(r)&=d'_1(r)+d'_2(r)+d'_3(r),
\end{align}
where
\begin{align}
  d'_1(r)&=-\lim_{\rho\rightarrow\infty}\frac{\log
    \prob(e_1')}{\log\rho}\label{eq:DMT21},\\
 d'_2(r)&=-\lim_{\rho\rightarrow\infty}\frac{\log
    \prob(e_2'| e_1')}{\log\rho},\label{eq:dmteq01}\\
\label{eq:dmteq}
  d'_3(r)&=-\lim_{\rho\rightarrow\infty}\frac{\log \prob(e_3'| e_2'
    ,e_1')}{\log\rho}.
\end{align}

Although the expressions above are in terms of error events, in the
remainder of this appendix the diversities are expressed in terms of
outage events instead of error events due to the fact that the
genie-aided modes are equivalent to  MISO channels and the codewords are
assumed to be long enough.

{\em Mode 1}, access mode (c), represents a parallel Rayleigh
channel. The outage of a parallel Rayleigh channel, $\prob({\cal O}_3
)$, is given by
\begin{align}
  \prob(&{\cal{O}}_3)= \prob({\cal O}_{31})+\prob({\cal O}_{32})+\prob({\cal O}_{33}),
\end{align}
where ${\cal O}_{31}, {\cal O}_{32}$ and ${\cal O}_{33}$ partition the
outage event ${\cal O}_3$ according to whether the first, the second, or
both direct links are in outage.
\begin{align}  \prob({\cal O}_{31})&=\prob \Big(\log(1+|h_{11}|^2\rho)<\frac{r}{2}\log\rho\Big)\prob \Big(\log(1+|h_{22}|^2)\rho\ge \frac{r}{2}\log\rho\Big)\label{eq:O31}\\
\prob({\cal O}_{32})&=\prob \Big(\log(1+|h_{11}|^2\rho)\ge\frac{r}{2}\log\rho\Big)\prob \Big(\log(1+|h_{22}|^2)\rho< \frac{r}{2}\log\rho\Big)\label{eq:O32}\\
\prob({\cal O}_{33})&=\prob \Big(\log(1+|h_{11}|^2\rho)<\frac{r}{2}\log\rho\Big)\prob \Big(\log(1+|h_{22}|^2)\rho< \frac{r}{2}\log\rho\Big).\label{eq:O33}
\end{align}
Therefore, in the asymptote of high SNR:
\begin{align}
\label{eq:out3SRC}
  P({\cal{O}}_3)
&\doteq  \rho^{{r/2}-1} + \rho^{{r/2}-1}+\rho^{r-2}\doteq  \rho^{{r/2}-1}
\end{align}
The unconditional DMT of the non-relayed mode 
\begin{equation}
  d'_1(r)=(1-\frac{r}{2})^+.
\label{eq:dmt3rd}
\end{equation}

 To calculate $d'_2(r)$ and $d'_3(r)$, we study the outage of the
respective access modes. We start by calculating the conditional
outage of Mode 3 and use the result to calculate the conditional
outage for Mode 2.
\begin{align}
\prob({\cal U}_3|{\cal U}_2,{\cal U}_1)
&=\prob \bigg(\Big\{\log\big(1+(|\hii |^2+|\hri |^2)\rho\big)<r\log\rho\Big\}\; \Big|\; \nonumber\\
&\qquad \Big\{|h_{jj}|^2<f_2^{-1}(R,|h_{r,d_j}|^2)\Big\}\; , \;\Big\{|h_{11}|^2<\frac{\rho^{r/2}-1}{\rho}\Big\},\Big\{
|h_{22}|^2<\frac{\rho^{r/2}-1}{\rho} \Big\} \bigg)\nonumber\\
&\doteq \prob \bigg(\Big\{|\hii |^2+|\hri |^2<\frac{\rho^r-1}{\rho}\Big\}\;\Big|\; \Big\{|\hii |^2
<\frac{\rho^{r/2}-1}{\rho}\Big\}\bigg),
\label{eq:condout}
\end{align}
where $i$ is the index of the source selected in Mode 3 and $j$ is the
index of the source selected in Mode 2. The channel gain $\gamma_{ii}
\defeq |\hii |^2$, conditioned on the event $B=\Big\{ |\hii
|<\frac{\rho^{r/2}-1}{\rho}\Big\}$ has the following conditional
distribution
\begin{equation}
\label{eq:exp_cond}
  f_{\gamma_{ii}|B}(x )=\begin{cases} \frac{
    e^{- x}}{1-e^{- \frac{\rho^{r/2}-1}{\rho}}} &
  x\le\frac{\rho^{r/2}-1}{\rho} \\ 0&x
  >\frac{\rho^{r/2}-1}{\rho} \end{cases}
\end{equation}

Defining $g_1(r,\rho)\defeq\frac{\rho^{r/2}-1}{\rho}$ and
$g_2(r,\rho)\defeq\frac{\rho^{r}-1}{\rho}$, the conditional probability
density function of $Z=|\hii |^2+|\hri |^2$ is 
\begin{equation}
  f_{Z|B}(z)=
  \begin{cases}
    \frac{z e^{-\lambda z}}{1-e^{- g_1(r,\rho)}}&z\le g_1(r,\rho)\\
    \frac{ g_1(r,\rho) e^{- z}}{1-e^{- g_1(r,\rho)}}&z>g_1(r,\rho).
  \end{cases}
\end{equation}
The probability of outage $\prob({\cal U}_3|{\cal U}_2,{\cal U}_1)$ can be
calculated as follows
\begin{align}
\prob({\cal U}_3|{\cal U}_2,{\cal U}_1)
&=\int_{0}^{g_1(r,\rho)} \frac{z
    e^{-z}}{1-e^{-g_1(r,\rho)}}dz+\int_{g_1(r,\rho)}^{g_2(r,\rho)}  \frac{g_1(r,\rho) e^{-z}}{1-e^{-g_1(r,\rho)}}dz\nonumber\\
&= \frac{1- e^{-g_1(r,\rho)}
-g_1(r,\rho)e^{-g_2(r,\rho)}}{1- e^{-g_1(r,\rho)}}\nonumber\\
&\doteq  1- \frac{\rho^{r/2-1} e^{-\rho^{r-1}}} {1- e^{-\rho^{r/2-1}}}\doteq  \rho^{r-1}.\label{eq:3rddmt}
\end{align}

To facilitate the analysis of the conditional outage of Mode 2, we
introduce a partition of ${\cal U}_1$. Define ${\cal V}$ as the event that {\em at least} one of the direct links can support half the
desired rate, i.e.  $\frac{r}{2}\log\rho$, and introduce:
\begin{equation}
{\cal V}_1 = {\cal V} \cap {\cal U}_1 \qquad\qquad {\cal V}_2 = \bar{\cal V} \cap {\cal U}_1
\end{equation}
Thus, ${\cal V}_1$ is the event that the non-relayed {\em Mode 1} is
in outage, and yet at least one of the two direct links can support at
least half the desired rate, i.e., $\frac{r}{2}\log\rho$.
\begin{align}
 \prob({\cal U}_2|{\cal U}_1)&=\frac{\prob({\cal U}_2,{\cal U}_1)}{\prob({\cal U}_1)}\nonumber\\
&=\frac{\prob({\cal U}_2|{\cal V}_1)\prob({\cal V}_1)+\prob({\cal U}_2|{\cal V}_2)\prob({\cal V}_2)}{\prob({\cal V}_1)+\prob({\cal V}_2)}\nonumber\\
&\doteq \frac{\rho^{2(r-1)}2\rho^{(\frac{r}{2}-1)}+\rho^{(r-1)}\rho^{2(\frac{r}{2}-1)}}{2\rho^{(\frac{r}{2}-1)}+\rho^{2(\frac{r}{2}-1)}}\nonumber\\
&\doteq \rho^{2(r-1)},\label{eq:2rddmt}
\end{align}
where $\prob({\cal V}_1)=\prob({\cal O}_{31})+\prob({\cal
  O}_{32})\doteq2\rho^{(r/2-1)}$ from Equation~(\ref{eq:O31})
and~(\ref{eq:O32}), $\prob({\cal V}_2)=\prob({\cal O}_{33})\doteq
\rho^{2(r/2-1)}$ from Equation~(\ref{eq:O33}). The probability of
$U_2$ conditioned on ${\cal V}_2$ is equivalent to
Equation~(\ref{eq:condout}) and hence $\prob({\cal U}_2|{\cal
  V}_2)\doteq \rho^{(r-1)}$. The conditional probability $\prob({\cal
  U}_2|{\cal V}_1)$ is given by
\begin{align}
  \prob({\cal U}_2|{\cal V}_1)=&\prob \bigg(\Big\{|\hii
  |^2+|\hri |^2<\frac{\rho^r-1}{\rho}\Big\}\;\Big|\; \Big\{|\hii |^2
  >\frac{\rho^{r/2}-1}{\rho}\Big\}\bigg).\label{eq:U2V1_1}
\end{align}
We notice that 
\begin{align}
  \prob \bigg(\Big\{|\hii
  |^2+&|\hri |^2<\frac{\rho^r-1}{\rho}\Big\}\bigg)\nonumber\\
=&\prob \bigg(\Big\{|\hii
  |^2+|\hri |^2<\frac{\rho^r-1}{\rho}\Big\}\;\Big|\; \Big\{|\hii |^2
  <\frac{\rho^{r/2}-1}{\rho}\Big\}\bigg)\prob\bigg(\Big\{|\hii |^2
  <\frac{\rho^{r/2}-1}{\rho}\Big\}\bigg)\nonumber\\
&+ \prob\bigg(\Big\{|\hii |^2+|\hri
|^2<\frac{\rho^r-1}{\rho}\Big\}\;\Big|\; \Big\{|\hii |^2
>\frac{\rho^{r/2}-1}{\rho}\Big\}\bigg)\prob\bigg(\Big\{|\hii |^2
>\frac{\rho^{r/2}-1}{\rho}\Big\}\bigg).\label{eq:U2V1_2}
\end{align}
At high SNR, using result from Equation~(\ref{eq:3rddmt}), Equation~(\ref{eq:U2V1_2}) leads to
\begin{align}
  \rho^{2(r-1)}\doteq\rho^{(r-1)}\rho^{(r/2-1)}+\prob\bigg(\Big\{|\hii |^2+|\hri
|^2<\frac{\rho^r-1}{\rho}\Big\}\;\Big|\; \Big\{|\hii |^2
>\frac{\rho^{r/2}-1}{\rho}\Big\}\bigg),\label{eq:U2V1_3}
\end{align}
where the random variable $|\hii |^2+|\hri |^2$ has Gamma
distribution. Using Equation~(\ref{eq:U2V1_1}) and~(\ref{eq:U2V1_3}),
one can see that $\prob({\cal U}_2|{\cal V}_1)=\rho^{2(r-1)}$.

  Equations~(\ref{eq:2rddmt}) and~(\ref{eq:3rddmt}) indicate that
\begin{align}
  d'_2(r)&=2(1-r)^+,\\ d'_3(r)&=(1-r)^+.
\end{align}
The DMT of the genie aided system is given by
\begin{align}
  d(r)&=(1-\frac{r}{2})^++2(1-r)^++(1-r)^+\\
&=
\begin{cases}
  4-\frac{7}{2}r, & 0\le r \le 1\\ (1-\frac{r}{2}), & 1 < r \le 2.
\end{cases}
\end{align}


\section{NAF Achievable DMT for the Shared Relay Channel}
\label{Appen:SRC_NAF}

The DMT of the NAF protocol for the shared relay channel will be
calculated according to the selection algorithm developed in
Appendix~\ref{Appen:SRC2}, which we invite the reader to review before
continuing with the present appendix. 

The overall diversity is governed by Equation~(\ref{eq:dmtadd}), and we
need to calculate $d_1'(r), d_2'(r), d_3'(r)$.

To begin with, the DMT of the non-relayed mode does not depend on the
relaying protocol, so there is no need to calculate it again: it is
$d'_1(r)= (1-\frac{r}{2})^+$ as calculated in
expression~(\ref{eq:dmt3rd}).

For calculating $d_3'(r)$, the equivalence of error and outage
analysis is nontrivial and will be relegated to
Appendix~\ref{Appen:NAF}. In this appendix we analyze the conditional
outage of {\em Mode 3}:
\begin{align}
\prob ({\cal U}_3|{\cal U}_2,{\cal U}_1 ) &\doteq \frac{1}{2} \prob\Big( I_1 <
R \;|\;\bar{\cal V},  I_2 < R \Big) + \frac{1}{2} \prob\Big( I_2 <
R \;|\;\bar{\cal V},  I_1 < R \Big)\nonumber\\
&\doteq \prob \Big ( I_1 < R \; \Big| \;
\Big\{ |h_{11}|^2<\frac{\rho^{r/2}-1}{\rho}\Big\}\Big ),
\label{eq:ModeThreeOutage}
\end{align}
where $I_1$ and $I_2$ are the instantaneous mutual information of the
simple relay channel for User 1 and User 2, respectively. The symmetry
arguments has been used to simplify the expression. We will use the
exponential order of channel gains, defined thus
\begin{equation}
\label{eq:exp_ord}
  v=-\lim_{\rho\rightarrow\infty}\frac{\log |h|^2}{\log\rho},
\end{equation}
where $v$ itself is a random variable. Recall that the conditional pdf
of the source-destination channel gain $|h_{11} |^2$, subject to $h_{11}$
not supporting rate $\frac{r}{2}\log \rho$, is given by
Equation~(\ref{eq:exp_cond}). The exponential order of this {\em
  conditional} random variable is denoted $v_1$ whose pdf can be
calculated as follows
\begin{equation}
  f(v_1)=
  \begin{cases}
    \ln\rho\;\rho^{-v_1}\frac{e^{-\rho^{-v_1}}}{1-e^{-\frac{\rho^{r/2}-1}{\rho}}}
    & v_1\ge1-\frac{r}{2},\\ 0 & v_1<1-\frac{r}{2}.
  \end{cases}
\end{equation}
As $\rho\rightarrow\infty$ we can show that
\begin{equation}
\label{eq:exp_ord_con}
  f(v_1)\doteq
  \begin{cases}
    \rho^{-v_1-(r/2-1)}& v_1\ge 1-\frac{r}{2},\\ 0 & v_1<1-\frac{r}{2}.
  \end{cases}
\end{equation}

Also, the channel gains $|h_{r1} |^2$ and $|h_{1r} |^2$ (exponentially
distributed, unconditioned) have exponential orders that are denoted
$v_2$ and $v_3$, respectively. Furthermore, the pdf of $v_1,v_2,v_3$ are
in turn characterized by their asymptotic exponential orders
$f(v_i)\doteq \rho^{-u_i}$, over their respective regions of support.

In a manner similar to~\cite{Azarian2005}, the outage region is more
conveniently addressed in the space of the exponential orders, i.e.
\begin{equation}
  O=\{(v_1,v_2,v_3) \; : \; I<r\log \rho\},
\end{equation}
We can now calculate:
\begin {align}
\label{eq:outgen}
\prob(I<r\log\rho)
&=\iiint_O f(v_1,v_2,v_3)\, dv_1\, dv_2\, dv_3\nonumber\\
&=\iiint_{O'}\log\rho\;\rho^{-v_1}\frac{e^{-\rho^{-v_1}}}{1-e^{-\frac{\rho^{r/2}-1}{\rho}}} \log\rho\;\rho^{-v_2}e^{-\rho^{-v_2}} \nonumber\\
&\qquad \times \log\rho\;\rho^{-v_3}e^{-\rho^{-v_3}} dv_1\, dv_2\, dv_3\nonumber\\
&\doteq \iiint_{O'}\rho^{-\sum u_i}dv_1\, dv_2\, dv_3\nonumber\\
&\doteq \rho^{-d_o},
\end{align}
where $O'$ is the intersection of $O$ and the support of
$f(v_1,v_2,v_3)$, and
\begin{align}
d_o&= {\;\underset{(v_1,v_2,v_3)\in
    O'}{\inf}\quad\sum^n_{j=1}u_i},\nonumber\\
&={\underset{(v_1,v_2,v_3)\in O'}
    {\inf}v_1+(r/2-1)+v_2+v_3}
\label{eq:AFout}
\end{align}
Following the same steps as those used in the proof of~\cite[Theorem
  2]{Azarian2005}, 
\begin{align}
\label{eq:AForgn}
    O'&=\Big\{(v_1,v_2,v_3)\in R^{3+},v_1\ge
    \big(1-\frac{r}{2}\big)\; , \; \Big[\max\Big((1-v_1),\frac{1}{2}(1-(v_2+v_3)\Big)\Big]^+<r\Big\}
    \end{align}
Solving~(\ref{eq:AFout}), we can show that 
\begin{equation}
 d_0= (1-2r)^+
\end{equation}
It remains to show that $d_3'(r)=d_0$, which will be done in
Appendix~\ref{Appen:NAF}.

For calculating $d_2'(r)$, we follow steps essentially similar to
those leading to Equation~(\ref{eq:2rddmt}), except this time we need
to make explicit the relationship between outage and error events.
\begin{align}
 \prob(e_2'|e_1') &\doteq \prob(e_2'|{\cal U}_1)
\label{eq:equiv1}\\
&=\frac{\prob(e_2',{\cal U}_1)}{\prob({\cal U}_1)}\nonumber\\
&=\frac{\prob(e_2'|{\cal V}_1)\prob({\cal V}_1)+\prob(e_2'|{\cal V}_2)\prob({\cal V}_2)}{\prob({\cal V}_1)+\prob({\cal V}_2)}\nonumber\\
&\doteq \frac{\rho^{(r-1)}\rho^{(2r-1)}2\rho^{(\frac{r}{2}-1)}+\rho^{(2r-1)}\rho^{2(\frac{r}{2}-1)}}{2\rho^{(\frac{r}{2}-1)}+\rho^{2(\frac{r}{2}-1)}}\label{eq:dominant}\\
&\doteq \rho^{-(1-r)^+-(1-2r)^+}.
\end{align}
where~(\ref{eq:equiv1}) is true because $e_1'$ is the error of a
non-relayed link therefore, with long codewords, it is exponentially
equivalent to the outage event ${\cal U}_1$.
Equation~(\ref{eq:dominant}) is derived by substituting the known
error exponents and noting that the third term is dominated by the
first two in both the numerator and denominator. Overall,
$d_2'(r)=(1-r)^++(1-2r)^+$ can be obtained.

To summarize, we have calculated $d_1'(r), d_2'(r)$ and $d_3(r)$.

\section{Relation of Outage and Error Events for the Shared Relay Channel}
\label{Appen:NAF}

In this appendix, we show that the outage and error events have the same
exponential order. The approach follows~\cite[Theorem 3]{Azarian2005}
and is adapted to the specific case at hand. We need to show $\prob(e)
\dot\le \prob({\cal O})$ and $\prob(e) \dot\ge \prob({\cal O})$. The
former is a straightforward application of~\cite[Lemma
  5]{Zheng2003}. For showing the latter inequality, note that
\begin{align}
\prob(e)&=\prob({\cal O})\prob(e|{\cal O})+\prob(e,\bar{\cal
  O})\nonumber\\ &\le\prob({\cal O})+\prob(e,\bar{\cal
  O})\nonumber\\ &\doteq \prob({\cal O})
\end{align}
where the last equation is valid whenever $\prob(e,\bar{\cal O}) \dot\le
\prob({\cal O})$, whose verification is the subject of the remainder of
this appendix.
The pairwise error probability conditioned on the channel coefficients
is given by
\begin{equation}
  P_{{\bf c}\rightarrow{\bf e}|h_{sd},h_{sr},h_{rd}}\le \det\bigg( {\bf I}_2+ \frac{1}{2}{\boldsymbol\Sigma}_s
  \boldsymbol\Sigma_n^{-1}\bigg)^{-\ell/2}
\end{equation}
where $\ell$ is the codebook codeword length and $\boldsymbol\Sigma_s$ and
$\boldsymbol\Sigma_n$ are the covariance matrices of the received signal and the
noise, respectively. The pair wise error probability is given by
\begin{equation}
  P_{{\bf c}\rightarrow{\bf e}|v_1,v_2,v_3}\dot\le\ \rho^{-\frac{l}{2}\max(2(1-v_1),1-(v_2+v_3))^+},
\end{equation}
where 
\begin{equation}
(v_1,v_2,v_3)\in \; R^{3+} \cap \Big\{v_1\ge\Big(1-\frac{r}{2}\Big)\Big\}.
\end{equation}
The total probability of error is
\begin{equation}
  P_{e|v_1,v_2,v_3}\dot\le
  \rho^{-\frac{l}{2}([\max(2(1-v_1),1-(v_2+v_3))]^+-2r)} ,
\end{equation}
The probability of error while no outage $\prob(e,\bar{\cal O})$ satisfies
\begin{align}
  \prob(e,\bar{\cal O}) & \dot\le \iiint_{O''}P_{e|v_1,v_2,v_3}\prob((v_1,v_2,v_3)\in\bar{\cal O})\, dv_1 \, dv_2 \, dv_3\nonumber\\
&= \iiint_{O''} 
\rho^{-\frac{l}{2}([\max(2(1-v_1),1-(v_2+v_3))]^+-2r)+v_1+(\frac{r}{2}-1)+v_2+v_3} \;
dv_1 \, dv_2 \, dv_3.
\end{align}
where $O'' = \{ (v_1,v_2,v_3) \in R^+ : (v_1,v_2,v_3) \notin O'\}$, the
area in the positive quadrant that is the complement of $O'$. Recall
that $O'$ is the outage region in the space of exponents, as defined
in~(\ref{eq:AForgn}). The integral is dominated by the minimum value of
the SNR exponent over $\bar{\cal O}$, i.e,
\begin{equation}
  \prob(e,\bar{\cal O})\dot\le\rho^{-d_1(r)},
\end{equation}
where
\begin{align}
  d_1(r)=\underset {v_1,v_2,v_3 \in O''}{\inf}&\frac{\ell}{2}\Big(\big[\max(2(1-v_1),1-(v_2+v_3))\big]^+
\; -2r\Big)+v_1+(r/2-1)+v_2+v_3.
\end{align}
Note that the multiplier of $\ell$ is positive throughout the region
$O''$. Now recall from the previous appendix that the outage probability is:
\begin{equation}
  \prob({\cal O}) \doteq \rho^{-d_0(r)},
\end{equation}
where
\begin{equation}
  d_0(r)=\underset{(v_1,v_2,v_3)\in O'}
    {\inf}v_1+(r/2-1)+v_2+v_3
\end{equation}
The expression for $d_1(r)$ has one extra term compared with $d_0(r)$
which, as mentioned above, is positive and can be made as large as
desired by choosing $\ell$ to be large enough. Therefore the condition
$\prob(e,\bar{\cal O}) \dot\le \prob({\cal O})$ is established,
leading to $\prob(e) \dot\le \prob({\cal O})$, which completes
the proof that the probability of error and outage events are
exponentially equivalent.


\section{DMT for DDF Opportunistic Shared Relay Channel}
\label{Appen:DDF}

We derive an achievable DMT for the DDF opportunistic shared relay
channel, employing the mode selection rule defined in
Appendix~\ref{Appen:SRC2}. The DMT is given by
Equations~(\ref{eq:dmtadd}), (\ref{eq:DMT21}), (\ref{eq:dmteq01})
and~(\ref{eq:dmteq}). The reader is referred to
Appendix~\ref{Appen:SRC2} for the definition of the access modes as well
as the selection rule.

The DMT for {\em Mode 1} is not affected by the relay and is given by
$d_1'(r)=(1-r/2)^+$, as seen in previous appendices. For {\em Mode 2}
one can employ the techniques of Appendix~\ref{Appen:SRC2} to show that
outage is dominated by the event of one link being in outage, hence
using results from~\cite{Azarian2005}, one can prove that
$$d_2'(r)=\
\begin{cases}
  2(1-r) & 0\le r\le \frac{1}{2}\\
  \frac{1-r}{r} & \frac{1}{2}\le 1
\end{cases}$$

To calculate $d_3'(r)$, we consider the conditional outage of {\em Mode
  3}; the equivalence of error and outage analysis can be shown in a
manner similar to Appendix~\ref{Appen:NAF} and~\cite{Azarian2005}, and
is omitted for brevity. In the following we directly derive diversity
from the outage events. The conditional outage of {\em Mode 3} was
calculated in Equation~(\ref{eq:ModeThreeOutage}):
\[
\prob ({\cal U}_3|{\cal U}_2,{\cal U}_1 ) \doteq \prob \Big ( I_1 < R \; \Big| \;
\Big\{ |h_{11}|^2<\frac{\rho^{r/2}-1}{\rho}\Big\}\Big )
\]
Given that
$|h_{11}|^2<\frac{\rho^{r/2}-1}{\rho}$, the exponential order of
$|h_{11}|^2$ is proved in~(\ref{eq:exp_ord_con}) to have the following
distribution at high SNR
\begin{equation}
  f(v_i)\doteq
  \begin{cases}
    \rho^{-v_i-(r/2-1)}& v_i\ge 1-\frac{r}{2},\\ 0 & v_i<1-\frac{r}{2}.
  \end{cases}
\end{equation}
The outage as shown in Equation~(\ref{eq:outgen}) is given by
\begin{equation}
\label{eq:outddf}
  \prob({\cal U}_3|{\cal U}_2,{\cal U}_1)\doteq \rho^{-d_3'(r)},
\end{equation}
where
\begin{align}
d_3'(r)={\underset{(v_1,v_2,v_3)\in O'} {\inf}v_1+(r/2-1)+v_2+v_3}.
\end{align}
Following the same steps as the proof of~\cite[Theorem 5]{Azarian2005},
the outage event $O'$ is defined as
\begin{align}
\label{eq:out_reg}
  O'&=\Big\{(v_1,v_2,v_3)\in R^{3+} \; ,\;v_1\ge(1-r/2)\; ,\;  t(1-v_1)^++
  (1-t)\big(1-\min(v_1,v_2)\big)^+\le r\Big\},
\end{align}
where $t$ is the listening-time ratio of the half-duplex relay, with $r\le
t\le 1$.

To get the DMT, we need to solve the optimization problem
of~(\ref{eq:outddf}),~(\ref{eq:out_reg}). 
 Solving the above optimizations and combining the results, the DMT is
given by
\begin{equation}
  d_3'(r)=
  \begin{cases}
    1-\frac{r}{1-r}\big(1-\frac{r}{2}\big),&0\le r \le 0.5\\
    \frac{(1-r)}{r}-\big(1-\frac{r}{2}\big),& 0.5<r\le 2-\sqrt2\\
    0, &2-\sqrt2 < r \le 1.
  \end{cases}
\end{equation}
Adding $d_1'(r)$, $d_2'(r)$ and $d_3'(r)$ completes the proof.


\section{DMT for CF Opportunistic Shared Relay Channel}
\label{Appen:CF}

The methods of this appendix closely follow~\cite{Yuksel2007}, with the
notable exception of implementing the effects of our selection algorithm
and the dependence between the nodes.

We use the selection criterion defined in Appendix~\ref{Appen:SRC2}, and
the DMT is given by Equations~(\ref{eq:dmtadd}), (\ref{eq:DMT21}),
(\ref{eq:dmteq01}) and (\ref{eq:dmteq}).  The DMT of non-relayed {\em
  Mode 3} is given by $d_1'(r)=(1-r/2)^+$, as seen several times
already, since it is not contingent on the relay protocol.

To calculate $d_2'(r)$ and $d_3'(r)$, we borrow the following result
from~\cite{Yuksel2007}. For the random half-duplex single-antenna
relay channel, the dynamic-state CF protocol is DMT optimal and by
random here we mean that the random binary state of the relay
(listen/transmit) is used as a channel input and used in designing
codebooks to convey information through the state of the relay.

For {\em Mode 2}, one can employ the techniques of
Appendix~\ref{Appen:SRC2} to show that outage is dominated by the
event of one link being in outage, hence using results
from~\cite{Yuksel2007}, one can prove that
\[
d_2'(r)=2(1-r)^+
\]

For {\em Mode 3}, the DMT is given by
\begin{align}
  d_3'(r)&=\max_t\min( d_{MAC}(r,t),d_{BC}(r,t)),
\end{align}
where
\begin{align}
   d_{BC}&=-\lim_{\rho\rightarrow
     \infty}\frac{\log\min_{p(x_s,x_r|q)}
     \prob\big(I_{BC}<r\log\rho|{\cal U}_2,{\cal U}_1\big)}{\log\rho},\nonumber\\
   d_{MAC}&=-\lim_{\rho\rightarrow\infty}\frac{\log\min_{p(x_s,x_r|q)}
     \prob\big(I_{MAC}<r\log\rho|{\cal U}_2,{\cal U}_1\big)}{\log\rho}\nonumber,
\end{align}
where $q$ represents the state of the relay (listening
vs. transmitting), $p(x_s,x_r|q)$ is the probability density of the
codebooks generated for the source and the relay, and $I_{BC}$ and
$I_{MAC}$ represent the total mutual information across the source
cutset and the destination cutset, respectively. It can be
shown~\cite{Yuksel2007} that
\begin{align}
I_{BC}\le&(1-t)\log(1+(|h_{s^*d^*}|^2+|h_{s^*r}|^2)\rho)+t\log(1+|h_{s^*d^*}|^2\rho)\nonumber\\
I_{MAC}\le&(1-t)\log(1+|h_{s^*d^*}|^2\rho)+t\log(1+(|h_{s^*d^*}|^2+|h_{rd^*}|^2)\rho)\nonumber
\end{align}
where $s^*$ and $d^*$ are the selected source and destination for Mode
3. Using the same technique as in Appendix~\ref{Appen:DDF}, we have
\begin{align}
  \prob\big(I_{BC}<r\log\rho|{\cal U}_2,{\cal U}_1\big)\doteq\rho^{-d_{BC}(r)},
\end{align}
where
\begin{equation}
  d_{BC}(r)={\underset{(v_1,v_3)\in O'}
    {\inf}v_1+(r/2-1)+v_3},
\end{equation}
and the outage event $O'$ is defined as
\begin{align}
  O'&=\Big\{(v_1,v_3)\in R^{2+},v_1\ge(1-r/2)\, , \,
(1-t)(1-v_1)^++
  t\big(1-\min(v_1,v_3)\big)^+\le r\Big\}.
\end{align}
Solving the optimization problem, the DMT for the source cutset
is given by
\begin{align}
  d_{BC}=
\begin{cases}
    1-\frac{r}{t}\Big( 1-\frac{1-t}{2}\Big)&t> \frac{1}{2},\;
    r\le\frac{1-(1-t)}{1-(1-t)/2}\\
 0 &t> \frac{1}{2},\; r>\frac{1-(1-t)}{1-(1-t)/2}\\
1-r\Big(\frac{1}{t}-\frac{1}{2}\Big)& t\le \frac{1}{2},\; r\le t\\
 \frac{1-r}{1-t}+\frac{r}{2}-1 & t\le
 \frac{1}{2},\;\frac{1-(1-t)}{1-(1-t)/2}\ge r>t\\
0&t\le \frac{1}{2},\; r>\frac{1-(1-t)}{1-(1-t)/2}.
  \end{cases}
\end{align}
  
Similarly, The DMT for the destination cutset is given by
\begin{align}
  d_{MAC}=
\begin{cases}
     1-\frac{r}{1-t}\Big( 1-\frac{t}{2}\Big)&t< \frac{1}{2},\;
     r\le\frac{1-t}{1-t/2}\\
     0 &t< \frac{1}{2},\;r>\frac{1-t}{1-t/2}\\
     1-r\Big(\frac{1}{1-t}-\frac{1}{2}\Big)& t\ge \frac{1}{2},\;
 r\le(1-t)\\ 
     \frac{1-r}{t}+\frac{r}{2}-1 & t\ge
     \frac{1}{2},\;\frac{1-t}{1-t/2}\ge r>(1-t)\\
     0&t\ge \frac{1}{2},\;r>\frac{1-t}{1-t/2}.
  \end{cases}
\end{align}
  
The two functions are equal at $t=\frac{1}{2}$ and that
gives the maximum DMT. The DMT is given by
\begin{align}
  d_3'(r)=\Big(1-\frac{3}{2}r\Big)^+.
\end{align}
Adding the DMT of the three modes completes the proof.

%

 \bibliographystyle{IEEEtran} \bibliography{IEEEabrv,mohamed}

\end{document}